\documentclass[english]{article}
\usepackage[T1]{fontenc}
\usepackage[latin9]{inputenc}
\usepackage{amsthm}
\usepackage{amsmath}
\usepackage{graphicx}
\usepackage{setspace}
\usepackage{amssymb}
\usepackage{esint}
\usepackage[numbers]{natbib}
\makeatletter
  \theoremstyle{plain}
  \newtheorem{lem}{Lemma}
\theoremstyle{plain}
\newtheorem{thm}{Theorem}
  \theoremstyle{plain}
  \newtheorem{prop}{Proposition}
  \theoremstyle{plain}
  \newtheorem{cor}{Corollary}

\usepackage{fullpage}
\date{}

\newcounter{hypA}
\newenvironment{condition}{\refstepcounter{hypA}\begin{itemize}
\item[({\bf A\arabic{hypA}})]}{\end{itemize}}

\newcounter{mycase}

\newcounter{mylemma}

\AtBeginDocument{
  
}

\makeatother

\usepackage{babel}

\begin{document}

\title{Sequential Monte Carlo samplers: error bounds and insensitivity to
initial conditions}

\author{Nick Whiteley%
\thanks{School of Mathematics, University of Bristol, University Walk, Bristol,
BS8 1TW, UK%
}}
\maketitle
\begin{abstract}
This paper addresses finite sample stability properties of sequential
Monte Carlo methods for approximating sequences of probability distributions.
The results presented herein are applicable in the scenario where
the start and end distributions in the sequence are fixed and the
number of intermediate steps is a parameter of the algorithm. Under
assumptions which hold on non-compact spaces, it is shown that the
effect of the initial distribution decays exponentially fast in the
number of intermediate steps and the corresponding stochastic error
is stable in $\mathbb{L}_{p}$ norm. 
\end{abstract}



\section{Introduction \label{sec:Introduction}}

Sequential Monte Carlo (SMC) methods are a class of stochastic algorithms
for approximating sequences of probability measures using a population
of $N$ particles. They have been adopted in a variety of application
domains, including rare event analysis \citep{smc:theory:DMG05},
statistical physics \citep{smc:meth:RS06}, optimal filtering \citep{smc:meth:DDG2001,hmm:theory:CMR05}
and computational statistics \citep{smc:meth:C02,smc:meth:DDJ06}.
Various theoretical properties of SMC methods have been studied, and
in various contexts, see amongst others \citep{smc:the:C04,smc:the:K05,smc:the:DM08,smc:the:vH09},
and the seminal work of \citet{smc:theory:Dm04}. Existing stability
results for SMC methods rely on very strong assumptions which typically
do not hold on non-compact spaces. This article is concerned with
establishing stability properties for a class of SMC algorithms primarily
motivated by those of \citet{smc:meth:DDJ06}, under assumptions which
do hold on non-compact spaces. The following example indicates a scenario
of interest.

\subsection{A motivating example}

Let $\pi$ be a distribution on some space $\mathsf{X}$, admitting
a strictly positive and bounded density with respect to some dominating
measure. For ease of presentation, let $\pi$ also denote this density
and let $\bar{\pi}$ denote the corresponding unnormalised density,
i.e. for some $Z>0$, $\pi(x)=\bar{\pi}(x)/Z$. For some $\underline{\gamma}\in(0,1)$,
let $\gamma:[0,1]\rightarrow[\underline{\gamma},1]$ be a non-decreasing
Lipschitz function with $\gamma(0)=\underline{\gamma}$ and $\gamma(1)=1$.
For $n\in\mathbb{N}$ let $\left\{ G_{n,k};0\leq k<n\right\} $ be
the collection of potential functions on $\mathsf{X}$ defined by
$G_{n,k}(x)=\bar{\pi}^{\gamma(\left(k+1\right)/n)}(x)/\bar{\pi}^{\gamma(k/n)}(x)$.
Let $\left\{ M_{n,k};1\leq k\leq n\right\} $ be a collection of ergodic
Markov kernels also on $\mathsf{X}$, where $M_{n,k}$ admits as its
invariant distribution the probability with density proportional to
$\bar{\pi}^{\gamma(k/n)}(x)$ (denoted by $\pi_{\gamma(k/n)}$). For
some initial distribution $\mu$, consider the sequence of probability
distributions $\left\{ \eta_{n,k};1\leq k\leq n\right\} $ defined,
for a test function $f$, by\begin{eqnarray}
\eta_{n,k}\left(f\right) & := & \frac{\mathbb{E}_{\mu}\left[\prod_{k=0}^{n-1}G_{n,k}(X_{n,k})f(X_{n,n})\right]}{\mathbb{E}_{\mu}\left[\prod_{k=0}^{n-1}G_{n,k}(X_{n,k})\right]},\label{eq:fk_intro}\end{eqnarray}
where $\mathbb{E}_{\mu}$ denotes expectation with respect to the
law of the inhomogeneous Markov chain $\left\{ X_{n,k};0\leq k\leq n\right\} $,
such that $X_{0}\sim\mu$ and $X_{n,k+1}|X_{n,k}=x_{n,k}\sim M_{n,k+1}\left(x_{n,k},\cdot\right)$.
The special interest in (\ref{eq:fk_intro}) is that when $\mu=\pi_{\underline{\gamma}}$,
then $\eta_{n,n}=\pi$, for any $n\geq1$ . Furthermore, in applications
$Z$ is unknown, the distributions of the form (\ref{eq:fk_intro})
cannot be computed exactly, and one aims to obtain an approximation
of $\pi$ by fixing $n$ then approximating each of the $\left\{ \eta_{n,k};0\leq k\leq n\right\} $
in turn, as follows. Let $N\in\mathbb{N}$ and let $\left\{ \zeta_{n,k};0\leq k\leq n\right\} $
be an inhomogeneous Markov chain, with each $\zeta_{n,k}=\left\{ \xi_{n,k}^{i};1\leq i\leq N\right\} $
an $N$-tuplet and with each $\xi_{n,k}^{i}$ valued in $\mathsf{X}$,
with $\left\{ \xi_{n,0}^{i};1\leq i\leq N\right\} $ independent and
of common distribution $\mu$. Given $\zeta_{n,k}$, the $\left\{ \xi_{n,k+1}^{i};1\leq i\leq N\right\} $
are independent, with $\xi_{n,k+1}^{i}$ drawn from \[
\frac{\sum_{j=1}^{N}G_{n,k}\left(\xi_{n,k}^{j}\right)M_{n,k+1}\left(\xi_{n,k}^{j},\cdot\right)}{\sum_{j=1}^{N}G_{n,k}\left(\xi_{n,k}^{j}\right)}.\]
The particle approximation measure at time step $k$ is $\eta_{n,k}^{N}:=\frac{1}{N}\sum_{i=1}^{N}\delta_{\xi_{n,k}^{i}}$,
and one takes $\eta_{n,n}^{N}(f):=\frac{1}{N}\sum_{i=1}^{N}f(\xi_{n,n}^{i})$
as an approximation of $\pi(f)$. The expectation terms in (\ref{eq:fk_intro})
are in the shape of Feynman-Kac formulae, and adopting the terminology
of \citet{smc:theory:Dm04} throughout the following, a general collection
of Markov kernels, potential functions, initial distribution and associated
$\left\{ \eta_{n,k}\right\} $ are referred to as constituting a Feynman-Kac
(FK) model.

The FK model described above has a notable structural characteristic:
due to fact that $\underbar{\ensuremath{\gamma}}$ and $\pi$ are
fixed and $\gamma(\cdot)$ is continuous, the potential functions
each become flat as $n\rightarrow\infty$ (note that this is not an
essential feature of the general FK models considered by \citet{smc:theory:Dm04},
nor is it the regime usually considered in the filtering scenario).
One might then conjecture due to this {}``flattening'' property
that the sequence of measures $\left\{ \eta_{n,k};0\leq k\leq n\right\} $
inherit ergodicity properties from the Markov kernels when $\mu\neq\pi_{\bar{\gamma}}$,
and $\eta_{n,n}(f)-\pi(f)$ goes to zero at some rate as $n\rightarrow\infty$.
Perhaps more adventurously, one might further conjecture that this
stability property is inherited by the corresponding particle system
so that $\mathbb{\bar{E}}_{\mu}\left[\left|\eta_{n,n}^{N}(f)-\pi(f)\right|^{p}\right]^{1/p}$
is controlled uniformly in $n$, and diminishes at the usual $\sqrt{N}$
rate, where $\bar{\mathbb{E}}_{\mu}$ denotes expectation w.r.t. the
law of the particle process initialised using $\mu$. 

The results presented in this paper allow it to be shown that this
is indeed the case under assumptions which are realistic in the context
of applications such as \citep{smc:meth:DDJ06} on non-compact spaces.
The contributions are to establish deterministic stability results
for a broad class of FK models, of which the example above is one
instance, and to provide $\mathbb{L}_{p}$ bounds for the corresponding
particle errors.

\subsection{Summary of results}

The present work is built upon generic assumptions about the FK model
structure, which can be loosely summarized as follows: 
\begin{itemize}
\item the Markov kernels $\left\{ M_{n,k}\right\} $ are geometrically ergodic,
with common Foster-Lyapunov drift function $V$ and associated constants,
and with a common minorization condition
\item the potential functions $\left\{ G_{n,k}\right\} $ are uniformly
bounded above, and are of the form $G_{n,k}(x)\propto\exp\left[-\frac{1}{n}U_{n,k}(x)\right]$,
for $U_{n,k}$ positive and bounded uniformly over $n$ and $k$ in
$V$-norm .
\item $f$ is bounded in $V$-norm
\end{itemize}
Stability properties of general FK semigroups are established in Theorem
\ref{thm:v_norm_bound} (section \ref{sec:Stability-of-model}). Note
here that, in contrast to the above example, no specific form is assumed
for $\left\{ U_{n,k}\right\} $ or their relations to the invariant
distributions of $\left\{ M_{n,k}\right\} $. Theorem \ref{thm:L_p_bound}
(section \ref{sec:-Error-Bounds}) provides $\mathbb{L}_{p}$ error
bounds for the corresponding particle systems, under additional assumptions
on the drift properties of the particle process.

For the reader's convenience Theorem \ref{thm:app} is now summarized
(for the precise statement see section \ref{sec:Application}), which
is an application to the example sketched above. Let $\mathsf{X}=\mathbb{R}^{d}$.
Then when $\pi$ has a sub-exponential density w.r.t. Lebesgue measure
with asymptotically regular contours, and when each $M_{n,k}$ is
a random walk Metropolis kernel of invariant distribution $\pi_{\gamma(k/n)}$,
under a suitable trade-off between $p\geq1$, $\alpha\in[0,1)$ and
$\underline{\gamma}$, there exists $\rho<1$ and constants $C_{1}$
and $C_{2}$ such that for any $f$ with $\left\Vert f\right\Vert _{V^{\alpha}}:=\sup_{x}\dfrac{\left|f(x)\right|}{V^{\alpha}(x)}<+\infty,$
any $N\geq1$ and $n\geq1$,\begin{eqnarray}
\bar{\mathbb{E}}_{\mu}\left[\left|\pi_{n}^{N}(f)-\pi(f)\right|^{p}\right]^{1/p} & \leq & \left\Vert f\right\Vert _{V^{\alpha}}\left(\frac{C_{1}}{\sqrt{N}}\left(\frac{1-\rho^{n}}{1-\rho}\right)+\rho^{n}C_{2}\mathbb{I}\left[\mu\neq\pi_{\underbar{\ensuremath{\gamma}}}\right]\right),\label{eq:app_result_front}\end{eqnarray}
where $\pi_{n}^{N}:=\eta_{n,n}^{N}$, is the particle approximation
as in section \ref{sec:Introduction}. The first term on the r.h.s.
of (\ref{eq:app_result_front}) corresponds to the stochastic errors
from the population of size $N$ interacting and mutating over $n$
times steps. The second term on the r.h.s. is a deterministic bias
which arises if the initial distribution is mis-specified, and this
bias is decays exponentially quickly in $n$.

\subsection{Existing work}

Despite the fact that the focus here is on models with the {}``flattening''
property, existing SMC stability results cannot be transferred directly
to the present scenario under realistic assumptions. They are reviewed
below for completeness.

\citet[Theorem 7.4.4]{smc:theory:Dm04} proved time-uniform $\mathbb{L}_{p}$
error bounds, under assumptions which in the present scenario would
take the form:
\begin{itemize}
\item there exist $\epsilon_{M}>0$, $m\geq0$ and $\epsilon_{G}>0$ such
for all $n$, $k$ and $x,y\in\mathsf{X}$\begin{eqnarray}
M_{n,k}\ldots M_{n,k+m}(x,\cdot) & \geq & \epsilon_{M}M_{n,k}\ldots M_{n,k+m}(y,\cdot)\label{eq:mixing}\\
G_{n,k}(x) & \geq & \epsilon_{G}G_{n,k}(y),\label{eq:G_bounded_above_and_below}\end{eqnarray}

\item $f$ is bounded.
\end{itemize}
Note that these results hold for fully general FK models, not necessarily
having the {}``flattening'' property (see also \citep{smc:the:DMM00,smc:the:DMG01,smc:the:DMD04,smc:the:LGO04,smc:the:K05,smc:the:CRDM11,smc:the:DMJDJ11}
for various results under the same type of assumption as one or both
of (\ref{eq:mixing})-(\ref{eq:G_bounded_above_and_below})). However,
these assumptions are very strong. Equation (\ref{eq:mixing}) is
stronger than uniform ergodicity of the m-step kernels, and typically
is not satisfied for the kernels of interest in \citep{smc:meth:DDJ06},
such as Metropolis-Hastings kernels on $\mathbb{R}^{d}$. For a toy
example which highlights the issue, consider the case that $\mathsf{X}=\mathbb{R}$
and for some probability measure $\nu$ on $\mathsf{X}$ dominated
by Lebesgue measure, take $M_{n,k}(x,\cdot)=a\delta_{x}(\cdot)+(1-a)\nu(\cdot)$.
It is easy to check that (\ref{eq:mixing}) is violated. Similarly
(\ref{eq:G_bounded_above_and_below}) is typically not satisfied in
the applications of interest on non-compact spaces and the assumption
that $f$ is bounded is then also rather restrictive. 

\citet{smc:the:OR05} and \citet{smc:the:HC08} used truncation approaches
to obtain stability results for particle filters in expectation over
the observation process, without mixing assumptions, but they respectively
introduce a rejection step into the particle algorithm (making its
computational cost random) and use restrictive assumptions about the
state-spaces involved, the hidden Markov model (HMM) and the particle
mutation kernels which are not realistic in the scenario of interest.
\citet{smc:the:vH09} proved uniform time-average consistency of particle
filters under tightness assumptions for a class of bounded functions.
Very recently, \citet{smc:the:DMJDJ11} have studied SMC methods in
which the resampling step is applied adaptively over time.

The stability of SMC methods has also been studied in the asymptotic
regime $N\rightarrow\infty$. \citet{smc:the:C04} established a CLT
for a broad class of SMC algorithms, and showed that under the same
type of strong mixing assumptions as in (\ref{eq:mixing})-(\ref{eq:G_bounded_above_and_below}),
the asymptotic (in $N$) variance associated with the rescaled stochastic
error can be bounded uniformly in $n$. \citet{smc:the:JD08} focused
on the asymptotic variance corresponding to the algorithms proposed
by \citet{smc:meth:DDJ06} for unbounded functions. They obtained
a bound on the asymptotic variance under realistic geometric ergodicity
assumptions which are the same as those considered in the present
work. They do not consider the {}``flattening'' regime in their
assumptions, and their bounds on the asymptotic variance are not time
uniform. Further details of the relationship between the approach
of \citet{smc:the:JD08} and some of the ideas in this paper are postponed
until section \ref{sec:Stability-of-model}. 

It is relevant also to mention the recent results of \citet{filter:the:KV08,filter:the:DFMP09}
on forgetting of initial conditions in HMM's (i.e. without particle
approximation) assuming a local Doeblin condition on the Markov kernels:
this also does not hold in the Metropolis-on-$\mathbb{R}^{d}$ scenario
of interest and so these results cannot be transferred directly (the
results of \citet{filter:the:DGLM10} do not require ergodicity of
the kernel, but their assumptions do involve the same minorization/majorization
structure). \citet{filter:the:DMR09} proposed a generic approach
to optimal filter stability without particle approximation using ideas
of coupling inhomogeneous Markov chains and a similar approach is
exploited here, further comment is delayed until the later sections.

The remainder of the paper is structured in the following manner.
Section \ref{sec:Definitions-and-Assumptions} specifies the general
form of the FK models in question, associated semigroups and particle
systems. Section \ref{sec:Stability-of-model} deals with the deterministic
stability of the sequences of measures arising from the FK models.
$\mathbb{L}_{p}$ error bounds for the stochastic errors of the particle
approximations are derived in Section \ref{sec:-Error-Bounds}. Section
\ref{sec:Application} applies the results to the case where $\mathsf{X}=\mathbb{R}^{d}$
and when the Markov kernels are of the random walk Metropolis variety.
The appendix contains a discussion of drift properties of the particle
system.

\section{Definitions \label{sec:Definitions-and-Assumptions}}

Consider a \textit{state space} $\mathsf{X}$ and an associated countably
generated $\sigma$-algebra $\mathcal{B}(\mathsf{X})$, Let $\mathcal{P}(\mathsf{X})$
be the collection of probability measures on $\left(\mathsf{X},\mathcal{B}(\mathsf{X})\right)$.
For a measure $\mu$ on $\left(\mathsf{X},\mathcal{B}(\mathsf{X})\right)$,
an integral kernel $M:\mathsf{X}\times\mathcal{B}(\mathsf{X})\rightarrow[0,\infty)$
and a function $f:\mathsf{X}\rightarrow\mathbb{R}$, define $\mu f:=\int_{\mathsf{X}}f(x)\mu(dx)$,
$Mf(x):=\int_{\mathsf{X}}M(x,dy)f(y)$ and $\mu M(\cdot)=\int_{\mathsf{X}}\mu(dx)P(x,\cdot)$
. The function which assigns $1$ to every point in $\mathsf{X}$
is also denoted by $1$, and the indicator function on a set $A$
is denoted by $\mathbb{I}_{A}$.

Let $W:\mathsf{X}\rightarrow[1,\infty)$ and $f:\mathsf{X}\rightarrow\mathbb{R}$
be two measurable functions, then define the norm\[
\parallel f\parallel_{W}=\underset{x\in\mathsf{X}}{\sup}\frac{|f(x)|}{W(x)}.\]
 and let $\mathcal{L}_{W}=\{f:\parallel f\parallel_{W}<\infty\}$.
Let $\mu$ be a signed measure on $\left(\mathsf{X},\mathcal{B}(\mathsf{X})\right)$,
then define the norm\begin{eqnarray*}
\parallel\mu\parallel_{W} & = & \sup_{\left|\phi\right|\leq W}\left|\mu(\phi)\right|.\end{eqnarray*}

\subsection{Feynman-Kac models and associated semigroups}

Let $\mu\in\mathcal{P}(\mathsf{X})$ and for each $n\in\mathbb{N}$
let $\left\{ M_{n,k};1\leq k\leq n\right\} $ be a collection of Markov
kernels, each kernel acting $\mathsf{X}\times\mathcal{B}(\mathsf{X})\rightarrow[0,1]$.
Let $\left\{ G_{n,k};0\leq k\leq n-1\right\} $ be a collection of
$\mathcal{B}(\mathsf{X})$-measurable, real-valued, strictly positive
and bounded functions on $\mathsf{X}$. The notation employed below
is directly inspired by that of \citet{smc:theory:Dm04}, with some
important modifications, primarily to the indexing, reflecting the
scenario of interest in which there is a different FK model for each
$n$. 

Next, for each $n\in\mathbb{N}$, let $\left\{ Q_{n,k};1\leq k\leq n\right\} $
be the collection of integral kernels defined by\[
Q_{n,k}(x,dy)=G_{n,k-1}(x)M_{n,k}(x,dy).\]
For each $n$ and $0\leq k\leq\ell\leq n$, let $M_{n,k:\ell}$ and
$Q_{n,k:\ell}$ the semigroups associated with the Markov kernels
$\left\{ M_{n,k}\right\} $ and the kernels $\left\{ Q_{n,k}\right\} $.
These semigroups are defined by\begin{eqnarray}
M_{n,k:\ell} & = & M_{n,k+1}M_{n,k+2}\ldots M_{n,\ell},\quad k<\ell\leq n,\nonumber \\
Q_{n,k:\ell} & = & Q_{n,k+1}Q_{n,k+2}\ldots Q_{n,\ell},\quad k<\ell\leq n,\label{eq:Q_semigroup}\end{eqnarray}
and $M_{n,k:k}=Q_{n,k:k}=Id,\quad0\leq k\leq n$. Next define the
collection of probability measures $\left\{ \eta_{n,k};0\leq k\leq n\right\} $
by\[
\eta_{n,k}(A)=\frac{\mu Q_{n,0:k}(A)}{\mu Q_{n,0:k}(1)},\;\;\; A\in\mathcal{B}(\mathsf{X}).\]
Let $\left\{ \Psi_{n,k};0\leq k<n\right\} $ and $\left\{ \Phi_{n,k};1\leq k\leq n\right\} $
be the collections of mappings, each mapping acting from $\mathcal{P}\left(\mathsf{X}\right)$
into $\mathcal{P}\left(\mathsf{X}\right)$, defined for any $\eta\in\mathcal{P}\left(\mathsf{X}\right)$
by\begin{eqnarray*}
\Psi_{n,k}(\eta)(dx)=\frac{G_{n,k}(x)}{\eta(G_{n,k})}\eta(dx), & \quad & \Phi_{n,k}(\eta)=\Psi_{n,k-1}(\eta)M_{n,k},\end{eqnarray*}
and for $0\leq k\leq\ell\leq n$ denote by $\Phi_{n,k:\ell}$ the
semigroup associated with the mappings $\left\{ \Phi_{n,k}\right\} $,
defined by\begin{eqnarray*}
\Phi_{n,k:\ell} & = & \Phi_{n,\ell}\circ\Phi_{n,\ell-1}\circ\ldots\circ\Phi_{n,k+1},\quad k<\ell\leq n,\end{eqnarray*}
and with the convention $\Phi_{n,k:k}=Id$. It is straightforward
to check that under these definitions, for any $0\leq k\leq\ell\leq n$,
$\eta\in\mathcal{P}\left(\mathsf{X}\right)$ and $A\in\mathcal{B}(\mathsf{X})$,\begin{eqnarray}
\Phi_{n,k:\ell}(\eta)(A) & = & \frac{\eta Q_{n,k:\ell}(A)}{\eta Q_{n,k:\ell}(1)}\label{eq:Phi_in_terms_ofQ}\end{eqnarray}
and in particular, \[
\Phi_{n,k:\ell}(\eta_{n,k})=\eta_{n,\ell}.\]
Lastly, let $\left\{ S_{n,k};1\leq k\leq n\right\} $ be the collection
of Markov kernels, each kernel acting $\mathsf{X}\times\mathcal{B}(\mathsf{X})\rightarrow[0,1]$,
defined by\begin{eqnarray*}
S_{n,k}(A) & = & \frac{M_{n,k}(Q_{n,k:n}(1)\mathbb{I}_{A})}{M_{n,k}(Q_{n,k:n}(1))}.\end{eqnarray*}

Under these definitions it is straightforward to check that, in line
with (\ref{eq:Phi_in_terms_ofQ}), we have the alternative description
of the mapping $\Phi_{n,k:n}$ in terms of the Markov kernels $\left\{ S_{n,k}\right\} $:
for any $0\leq k\leq n$ and any $\eta\in\mathcal{P}(\mathsf{X})$\begin{eqnarray}
\Phi_{n,k:n}(\eta)(A) & = & \frac{\eta\left(Q_{n,k:n}(1)S_{n,k+1}\ldots S_{n,n}(A)\right)}{\eta\left(Q_{n,k:n}(1)\right)}.\label{eq:transport_in_terms_of_S}\end{eqnarray}

Consider the following assumption.

\begin{condition}\label{hyp:G_bounded-1}There exists a finite constant
$\overline{G}_{\mathsf{X}}$ such that for all $n\in\mathbb{N}$ and
$0\leq k\leq n-1$,\begin{eqnarray*}
G_{n,k}(x) & \leq & \overline{G}_{\mathsf{X}},\quad\quad\forall x\in\mathsf{X}.\end{eqnarray*}

\end{condition}

When assumption (A\ref{hyp:G_bounded-1}) holds, for $0\leq k\leq n-1$
define $\widetilde{G}_{n,k}:\mathsf{X}\rightarrow(0,1]$ by

\[
\widetilde{G}_{n,k}(x)=\frac{G_{n,k}(x)}{\overline{G}_{\mathsf{X}}},\]
 and correspondingly,\begin{eqnarray*}
\widetilde{Q}_{n,k}(x,dy) & = & \widetilde{G}_{n,k-1}(x)M_{n,k}(x,dy).\end{eqnarray*}
Also let $\widetilde{Q}_{n,k:\ell}$ denote the semigroup associated
with the kernels $\widetilde{Q}_{n,k}$ in the same manner as (\ref{eq:Q_semigroup}),
with the same convention $\widetilde{Q}_{n,k:k}=Id$. Furthermore
under (A\ref{hyp:G_bounded-1}), let $U_{n,k}:\mathsf{X}\rightarrow[0,\infty)$
be defined by $U_{n,k}(x)=-n\log\widetilde{G}_{n,k}(x)$.

As stated in section \ref{sec:Introduction}, the work presented here
is primarily motivated by the models and algorithms considered in
\citep{smc:meth:DDJ06}. However, the {}``backward'' kernel structure
which they consider is not introduced here as it is not essential
for our purposes. A specific example is given in section \ref{sec:Application}
and at that point comment on how this fits with the framework of \citet{smc:meth:DDJ06}
is provided.

\subsection{Particle systems}

An explicit construction of the probability space for the particle
systems is not provided here, but this can be carried out by canonical
methods, see for example \citep[ Chapter 3]{smc:theory:Dm04} and
should be clear from the following symbolic description. Fix $N\in\mathbb{N}$,
$n\in\mathbb{N}$, and for $0\leq k\leq n$ let $\zeta_{n,k}^{(N)}:=\left\{ \mbox{\ensuremath{\xi}}_{n,k}^{\left(N,i\right)};1\leq i\leq N\right\} $,
where each $\mbox{\ensuremath{\xi}}_{n,k}^{\left(N,i\right)}$ is
valued in $\mathsf{X}$. Denote $\eta_{n,k}^{N}:=\dfrac{1}{N}\sum_{i=1}^{N}\delta_{\xi_{n,k}^{\left(N,i\right)}}$.
For $1\leq i\leq N$ and $1\leq k\leq n$ let $\mathcal{F}_{n,k}^{(N,i)}:=\sigma(\zeta_{n,0}^{(N)},\ldots,\zeta_{n,k-1}^{(N)},\mbox{\ensuremath{\xi}}_{n,k}^{\left(N,1\right)},...,\mbox{\ensuremath{\xi}}_{n,k}^{\left(N,i\right)})$
and $\mathcal{F}_{n,0}^{(N,i)}:=\sigma(\mbox{\ensuremath{\xi}}_{n,0}^{\left(N,1\right)},...,\mbox{\ensuremath{\xi}}_{n,0}^{\left(N,i\right)})$.
The generations of the particle system $\left\{ \zeta_{n,k}^{(N)};0\leq k\leq n\right\} $
form a non-homogeneous Markov chain: for $\mu\in\mathcal{P}(\mathsf{X})$,
the law of this chain is denoted by $\mathbb{P}_{\mu}$ and has transitions
given in integral form by:

\begin{eqnarray*}
\mathbb{P}_{\mu}\left(\zeta_{n,0}^{(N)}\in dx\right) & = & \prod_{i=1}^{N}\mu(dx^{i}),\\
\\\mathbb{P}_{\mu}\left(\left.\zeta_{n,k}^{(N)}\in dx\right|\zeta_{n,k-1}^{(N)}\right) & = & \prod_{i=1}^{N}\Phi_{n,k}\left(\eta_{n,k-1}^{N}\right)(dx^{i}),\quad1\leq k\leq n,\end{eqnarray*}
where $dx=d\left(x^{1},\ldots x^{N}\right)$. Denote by $\bar{\mathbb{E}}_{\mu}$
the expectation corresponding to $\mathbb{P}_{\mu}$. It is easy to
check that for any $0\leq k\leq n$ and$1\leq i\leq N$,\begin{eqnarray}
\bar{\mathbb{E}}_{\mu}\left[\mathbb{I}_{A}\left(\mbox{\ensuremath{\xi}}_{n,k}^{\left(N,i\right)}\right)\right] & = & \bar{\mathbb{E}}_{\mu}\left[\eta_{n,k}^{N}(A)\right],\quad A\in\mathcal{B}\left(\mathsf{X}\right).\label{eq:exch}\end{eqnarray}

\section{Stability of the deterministic measures\label{sec:Stability-of-model}}

This section is concerned with stability properties of the sequences
of Markov kernels $\left\{ S_{n,k}\right\} $ and operators $\left\{ \Phi_{n,k}\right\} $.
The approach is to identify non-homogeneous Foster-Lyapunov drift
functions and minorization conditions which arise quite naturally
from the structure of the FK model and then to employ the quantitative
bounds of \citet{mc:theory:DMR:04}. Compared to \citep{smc:the:JD08},
in the present work the general structure of FK models is exploited
more directly and in a way which is fruitful when the models satisfy
assumption (A\ref{hyp:norm_const_bound_below}) below. \citet[Section 5]{filter:the:DMR09}
identified drift functions and coupling sets for related operators
in some specific HMM's; the present work is concerned with a general
FK model structure. The first main idea of this section is illustrated
by the following assumption and lemma.

\begin{condition}\label{hyp:fosterlyapunov_prelim}

There exists $\lambda\in[0,1)$, a function $V:\mathsf{X}\rightarrow[1,\infty)$,
$\varepsilon\in(0,1]$, $b\in(0,\infty)$, $\mathsf{C}\in\mathcal{B}\left(\mathsf{X}\right)$
and a probability measure $\nu\in\mathcal{P}(\mathsf{X})$, such that
\begin{eqnarray}
\inf_{n\geq1}\inf_{1\leq k\leq n}M_{n,k}(x,A) & \geq & \varepsilon\cdot\nu(A),\;\;\;\forall A\in\mathcal{B}(\mathsf{X}),\;\forall x\in\mathsf{C},\label{eq:foster_lyap1-1}\\
\sup_{n\geq1}\sup_{1\leq k\leq n}M_{n,k}V(x) & \leq & \lambda V(x)+b\mathbb{I}_{\mathsf{C}}(x),\;\;\;\forall x\in\mathsf{X}.\label{eq:foster_lyap2-1}\end{eqnarray}

\end{condition}
\begin{lem}
\label{lem:non_hom_drift_prelim} Assume (A\ref{hyp:G_bounded-1})-(A\ref{hyp:fosterlyapunov_prelim}).
Then for each $n\in\mathbb{N}$ and $1\leq k\leq n$\textup{,}\begin{eqnarray}
S_{n,k}(x,\cdot) & \geq & \epsilon_{n,k}\nu_{n,k}(\cdot),\quad\forall x\in\mathsf{C},\label{eq:minor-1}\\
\nonumber \\S_{n,k}V_{n,k}(x) & \leq & \lambda V_{n,k-1}(x)+b_{n,k-1}\mathbb{I}_{\mathsf{C}}(x),\quad\forall x\in\mathsf{X},\label{eq:drift-1-1}\end{eqnarray}
where\begin{eqnarray}
\epsilon_{n,k}:=\varepsilon\cdot\nu\left(\widetilde{Q}_{n,k:n}(1)\right), & \quad & b_{n,k}:=\frac{b}{\varepsilon\cdot\nu\left(\widetilde{Q}_{n,k:n}(1)\right)},\label{eq:eps_n_p_prelim}\end{eqnarray}
and $\nu_{n,k}\in\mathcal{P}\left(\mathsf{X}\right)$ and $V_{n,k}:\mathsf{X}\rightarrow[1,\infty)$
are defined by\begin{eqnarray}
\nu_{n,k}(A) & := & \frac{\nu\left(\widetilde{Q}_{n,k:n}(1)\mathbb{I}_{A}\right)}{\nu\left(\widetilde{Q}_{n,k:n}(1)\right)},\quad A\in\mathcal{B}\mathsf{\left(X\right),}\label{eq:nu_n_p_prelim}\\
\nonumber \\V_{n,k}(x) & := & \frac{V(x)}{M_{n,k+1}\left(\widetilde{Q}_{n,k+1:n}(1)\right)(x)},\quad x\in\mathsf{X,}\quad0\leq k<n,\label{eq:V_n_p_prelim}\end{eqnarray}
$V_{n,n}:=V$, with $\lambda$, $V$, $b$, $\nu$, $\varepsilon$
and $\mathsf{C}$ as in (A\ref{hyp:fosterlyapunov_prelim}).\end{lem}
\begin{proof}
Noting that $0<\widetilde{Q}_{n,k:n}(1)(x)\leq1$ for all $x\in\mathsf{X}$,
we have\begin{eqnarray*}
S_{n,k}(x,A) & =\dfrac{M_{n,k}\left(\widetilde{Q}_{n,k:n}(1)\mathbb{I}_{A}\right)(x)}{M_{n,k}\left(\widetilde{Q}_{n,k:n}(1)\right)(x)} & \geq M_{n,k}\left(\widetilde{Q}_{n,k:n}(1)\mathbb{I}_{A}\right)(x)\\
 &  & \geq\varepsilon\cdot\nu\left(\widetilde{Q}_{n,k:n}(1)\right)\frac{\nu\left(\widetilde{Q}_{n,k:n}(1)\mathbb{I}_{A}\right)}{\nu\left(\widetilde{Q}_{n,k:n}(1)\right)},\quad x\in\mathsf{C}.\end{eqnarray*}
and\begin{eqnarray}
S_{n,k}V_{n,k}(x) & =\dfrac{M_{n,k}\left(\widetilde{G}_{n,k}V\right)(x)}{M_{n,k}\left(\widetilde{Q}_{n,k:n}(1)\right)(x)} & \leq\lambda\frac{V(x)}{M_{n,k}\left(\widetilde{Q}_{n,k:n}(1)\right)(x)}+b\frac{\mathbb{I}_{\mathsf{C}}(x)}{M_{n,k}\left(\widetilde{Q}_{n,k:n}(1)\right)(x)}\nonumber \\
 &  & \leq\lambda V_{n,k-1}(x)+b\frac{\mathbb{I}_{\mathsf{C}}(x)}{\varepsilon\cdot\nu\left(\widetilde{Q}_{n,k:n}(1)\right)},\quad\forall x\in\mathsf{X},\label{eq:uni_var_nh_drift-1}\end{eqnarray}
recalling the convention $\widetilde{Q}_{n,n:n}(1)=1$. 
\end{proof}
The minorization and drift conditions (\ref{eq:minor-1})-(\ref{eq:drift-1-1})
pave the way to establishing the stability properties of the operators
$\left\{ \Phi_{n,k}\right\} $. In order to move further we introduce
assumption (A\ref{hyp:fosterlyapunov}) below, which is a stricter
version of (A\ref{hyp:fosterlyapunov_prelim}). The extra structure
of (A\ref{hyp:fosterlyapunov}) allows the construction in Proposition
\ref{pro:bi_var_drift} of bi-variate drift and minorization conditions.
Consideration of the case in which the small set arises as a sub-level
set of the drift function $V$ is a standard and generic approach
to constructing bi-variate drift conditions from their uni-variate
counter-parts, see for example \citep{mcmc:the:ABD01,mc:theory:DMR:04}.
Furthermore, assumption (A\ref{hyp:fosterlyapunov}) allows the level
in question to be chosen in a very flexible way, and this property
is exploited in Proposition \ref{pro:bi_var_drift} when dealing with
the specific structure arising from the kernels $\left\{ S_{n,k}\right\} $.
Verification of (A\ref{hyp:fosterlyapunov}) in a particular application
is provided in section \ref{sec:Application}. Assumption (A\ref{hyp:norm_const_bound_below})
allows the bounding of non-homogeneous minorization and drift constants.
A generic approach to verifying (A\ref{hyp:norm_const_bound_below})
is presented at the end of this section, where the connection with
flattening property mentioned in the introduction is made more explicit.

\begin{condition}\label{hyp:fosterlyapunov}There exists $d_{0}\geq1$,
$\lambda\in[0,1)$, a function $V:\mathsf{X}\rightarrow[1,\infty)$,
and for all $d\in[d_{0},+\infty)$, there exists $\varepsilon_{d}\in(0,1]$,
$b_{d}\in(0,\infty)$ and $\nu_{d}\in\mathcal{P}(\mathsf{X})$ such
that $\nu_{d}\left(\mathsf{C}_{d}\right)>0$, $\nu_{d}(V)<+\infty$,
\begin{eqnarray}
\inf_{n\geq1}\inf_{1\leq k\leq n}M_{n,k}(x,A) & \geq & \varepsilon_{d}\cdot\nu_{d}(A),\;\;\;\forall A\in\mathcal{B}(\mathsf{X}),\;\forall x\in\mathsf{C}_{d},\label{eq:foster_lyap1}\\
\sup_{n\geq1}\sup_{1\leq k\leq n}M_{n,k}V(x) & \leq & \lambda V(x)+b_{d}\mathbb{I}_{\mathsf{C}_{d}}(x),\;\;\;\forall x\in\mathsf{X},\label{eq:foster_lyap2}\end{eqnarray}
where $\mathsf{C}_{d}:=\left\{ x:V(x)\leq d\right\} $. 

\end{condition}

\begin{condition}\label{hyp:norm_const_bound_below} Whenever (A\ref{hyp:G_bounded-1})
and (A\ref{hyp:fosterlyapunov}) hold, for any $\mu\in\mathcal{P}\left(\mathsf{X}\right)$
with $\mu(V)<+\infty$, and $d\in[d_{0},+\infty)$ there exists a
positive and finite constant $K\left(\mu,\lambda,V,b_{d}\right)$
such that\begin{eqnarray*}
\inf_{n\geq1}\inf_{0\leq k\leq n}\mu\left(\widetilde{Q}_{n,k:n}(1)\right) & \geq & K\left(\mu,\lambda,V,b_{d}\right),\end{eqnarray*}
where $\lambda$, $V$ and $d_{0}$ are as in (A\ref{hyp:fosterlyapunov}).\end{condition}

The main result of this section is now presented.
\begin{thm}
\label{thm:v_norm_bound}Assume (A\ref{hyp:G_bounded-1}), (A\ref{hyp:fosterlyapunov})
and (A\ref{hyp:norm_const_bound_below}). Then for $\alpha\in(0,1]$,
there exists $\rho\in(\lambda,1)$ and a finite constant $M$ such
that for any $\mu,\mu'\in\mathcal{P}(\mathsf{X})$, any $n\in\mathbb{N}$
and any $1\leq k\leq n$,\begin{eqnarray}
\left\Vert \mu S_{n,k}\ldots S_{n,n}-\mu'S_{n,k}\ldots S_{n,n}\right\Vert _{V^{\alpha}} & \leq & M\rho^{n-k+1}\left[\mu\left(V_{n,k-1}^{\alpha}\right)+\mu'\left(V_{n,k-1}^{\alpha}\right)\right],\label{eq:theorem_v_norm_bounds_statement}\end{eqnarray}
 where $V_{n,k-1}$ is as given in equation (\ref{eq:V_n_p_prelim}),
and consequently, for $0\leq k\leq n,$\begin{eqnarray}
\left\Vert \Phi_{n,k:n}\left(\mu\right)-\Phi_{n,k:n}\left(\mu'\right)\right\Vert _{V^{\alpha}} & \leq & M\rho^{n-k}\left[\frac{\mu\left(\widetilde{G}_{n,k}V^{\alpha}\right)}{\mu\left(\widetilde{Q}_{n,k:n}(1)\right)}+\frac{\mu'\left(\widetilde{G}_{n,k}V^{\alpha}\right)}{\mu'\left(\widetilde{Q}_{n,k:n}(1)\right)}\right].\label{eq:theorem_v_norma_bounds_Phi_state}\end{eqnarray}

\end{thm}
The proof of Theorem \ref{thm:v_norm_bound} is postponed. It involves
the bi-variate drift functions identified in the following proposition.
\begin{prop}
\label{pro:bi_var_drift}Assume (A\ref{hyp:G_bounded-1}) and (A\ref{hyp:fosterlyapunov})
and let $V$, $d_{0}$ and $\lambda$ be as in (A\ref{hyp:fosterlyapunov}).
Then for all $d\geq d_{0}$, and all $\bar{\lambda}\in(\lambda,1)$
there exists $\bar{d}\geq d$, and for each $n\in\mathbb{N},$ there
exist 

$\quad$

$\quad\quad$1) a collection of functions $\left\{ \bar{V}_{n,k};0\leq k\leq n\right\} $,
with each $\bar{V}_{n,k}:\mathsf{X}\times\mathsf{X}\rightarrow[1,\infty)$
and 

$\quad\quad$$\quad$ $\bar{V}_{n,n}(x,x')=\frac{1}{2}\left[V(x)+V(x')\right]$;

$\quad\quad$2) collections $\left\{ \bar{\epsilon}_{n,k};1\leq k\leq n\right\} $,
and $\left\{ \overline{b}_{n,k};0\leq k\leq n-1\right\} $ depending
on $d$ and $\bar{d}$, with each

$\quad\quad$$\quad$ $\bar{\epsilon}_{n,k}\in(0,1]$ and each $\bar{b}_{n,k}\in(0,\infty)$; 

$\quad\quad$3) a collection of probability measures $\left\{ \bar{\nu}_{n,k};1\leq p\leq n\right\} $,
depending on $\bar{d}$, with each

$\quad\quad$$\quad$ $\bar{\nu}_{n,k}\in\mathcal{P}(\mathsf{X})$;

$\quad$ 

such that for\textup{ $1\leq k\leq n$,}

\begin{eqnarray}
S_{n,k}(x,\cdot)\wedge S_{n,k}(x',\cdot) & \geq & \bar{\epsilon}_{n,k}\bar{\nu}_{n,k}(\cdot),\quad\forall(x,x')\in\bar{\mathsf{C}}_{\bar{d}},\label{eq:bi_var_minor}\\
\nonumber \\S_{n,k}^{*}\bar{V}_{n,k}(x,x') & \leq & \bar{\lambda}\bar{V}_{n,k-1}(x,x')+\overline{b}_{n,k-1}\mathbb{I}_{\bar{\mathsf{C}}_{\bar{d}}}(x,x'),\quad\forall(x,x')\in\mathsf{X}\times\mathsf{X},\label{eq:bi_var_drift-1}\end{eqnarray}
where $\left\{ \bar{\mathsf{C}}_{\bar{d}}:=(x,x'):V(x)\leq\bar{d},V(x')\leq\bar{d}\right\} $,\textup{
$S_{n,k}^{*}:\mathsf{X}\times\mathsf{X}\times\mathcal{B}(\mathsf{X}\times\mathsf{X})\rightarrow[0,1]$
is defined by \begin{eqnarray*}
S_{n,k}^{*}\left(\left(x,x'\right),d(y,y')\right) & = & \begin{cases}
S_{n,k}(x,dy)S_{n,k}(x',dy'), & (x,x')\notin\bar{\mathsf{C}}_{\bar{d}},\\
\bar{R}_{n,k}\left(\left(x,x\right),d(y,y')\right), & (x,x')\in\bar{\mathsf{C}}_{\bar{d}}\end{cases}\end{eqnarray*}
}and $\bar{R}_{n,k}:\mathsf{X}\times\mathsf{X}\times\mathcal{B}(\mathsf{X}\times\mathsf{X})\rightarrow[0,1]$
is defined by

\textup{\begin{eqnarray*}
\bar{R}_{n,k}\left(\left(x,x\right),d(y,y')\right) & = & \frac{1}{(1-\bar{\epsilon}_{n,k})^{2}}\left(S_{n,k}(x,dy)-\bar{\epsilon}_{n,k}\bar{\nu}_{n,k}(dy)\right)\left(S_{n,k}(x',dy')-\bar{\epsilon}_{n,k}\bar{\nu}_{n,k}(dy')\right).\end{eqnarray*}
}

Furthermore,\textup{\begin{eqnarray}
\inf_{n\geq1}\inf_{1\leq k\leq n}\bar{\epsilon}_{n,k}>0, & \text{ and } & \sup_{n\geq1}\sup_{0\leq k\leq n-1}\bar{b}_{n,k}<+\infty.\label{eq:eps_n_p_bounded_below}\end{eqnarray}
}\end{prop}
\begin{proof}
Throughout the proof, expressions featuring indices $n$ and $k$
hold for all $n\in\mathbb{N}$ and $1\leq k\leq n$, unless stated
otherwise. 

Fix $d\geq d_{0}$ and $\bar{\lambda}\in(\lambda,1)$. Then let $\varepsilon_{d}$,
$b_{d}$ and $\nu_{d}$ be the corresponding constants and minorizing
measure from (A\ref{hyp:fosterlyapunov}). Let $K\left(\nu_{d},\lambda,V,b_{d}\right)$
be the constant of assumption (A\ref{hyp:norm_const_bound_below})
corresponding to $\nu_{d}$ and $b_{d}$, and set \begin{eqnarray*}
\bar{d} & := & \left[\frac{b_{d}}{\varepsilon_{d}\left(\bar{\lambda}-\lambda\right)K\left(\nu_{d},\lambda,V,b_{d}\right)}-1\right]\vee d.\end{eqnarray*}
Then for equation (\ref{eq:bi_var_minor}), under assumption (A\ref{hyp:fosterlyapunov}),
there exists $\varepsilon_{\bar{d}}$ and $\nu_{\bar{d}}$ such that\begin{eqnarray*}
S_{n,k}(x,A) & \geq & \varepsilon_{\bar{d}}\cdot\nu_{\bar{d}}\left(\widetilde{Q}_{n,k:n}(1)\mathbb{I}_{A}\right)=\varepsilon_{\bar{d}}\cdot\nu_{\bar{d}}\left(\widetilde{Q}_{n,k:n}(1)\right)\frac{\nu_{\bar{d}}\left(\widetilde{Q}_{n,k:n}(1)\mathbb{I}_{A}\right)}{\nu_{\bar{d}}\left(\widetilde{Q}_{n,k:n}(1)\right)},\end{eqnarray*}
for all $x\in\mathsf{C}_{\bar{d}}$. Then setting \begin{eqnarray}
\bar{\epsilon}_{n,k} & := & \varepsilon_{\bar{d}}\cdot\nu_{\bar{d}}\left(\widetilde{Q}_{n,k:n}(1)\right),\quad\bar{\nu}_{n,k}(A):=\frac{\nu_{\bar{d}}\left(\widetilde{Q}_{n,k:n}(1)\mathbb{I}_{A}\right)}{\nu_{\bar{d}}\left(\widetilde{Q}_{n,k:n}(1)\right)},\quad A\in\mathcal{B}\left(\mathsf{X}\right),\label{eq:eps_n_p_bounded_below_proof1}\end{eqnarray}
establishes (\ref{eq:bi_var_minor}). The first part of (\ref{eq:eps_n_p_bounded_below})
is an immediate consequence of (A\ref{hyp:norm_const_bound_below})
and (\ref{eq:eps_n_p_bounded_below_proof1}):

\begin{equation}
\varepsilon_{\bar{d}}\cdot\nu_{\bar{d}}\left(\widetilde{Q}_{n,k:n}(1)\right)\geq\varepsilon_{\bar{d}}K\left(\nu_{\bar{d}},\lambda,V,b_{\bar{d}}\right).\label{eq:eps_n_p_bouned_below_proof2}\end{equation}
Let $\left\{ V_{n,k}\right\} $ be as defined in (\ref{eq:V_n_p_prelim}).
Consider the collection of bi-variate drift functions $\left\{ \bar{V}_{n,k};0\leq k\leq n\right\} $,
with each $\bar{V}_{n,k}:\mathsf{X}\times\mathsf{X}\rightarrow[1,\infty)$
defined by\begin{eqnarray}
\bar{V}_{n,k}(x,x') & := & \frac{1}{2}\left[V_{n,k}(x)+V_{n,k}(x')\right].\label{eq:bi_var_drift_in_terms_of_uni}\end{eqnarray}
We now proceed to establish the bi-variate drift condition of equation
(\ref{eq:bi_var_drift-1}). First, following the same arguments as
in the proof of Lemma \ref{lem:non_hom_drift_prelim}, under (A\ref{hyp:fosterlyapunov}),
\begin{eqnarray}
S_{n,k}V_{n,k}(x) & \leq & \lambda V_{n,k-1}(x)+b_{d}\frac{\mathbb{I}_{\mathsf{C}_{d}}(x)}{\varepsilon_{d}\cdot\nu_{d}\left(\widetilde{Q}_{n,k:n}(1)\right)},\quad\forall x\in\mathsf{X}.\label{eq:uni_var_nh_drift}\end{eqnarray}
From (\ref{eq:uni_var_nh_drift}), for $(x,x')\notin\bar{\mathsf{C}}_{\bar{d}}$
we have\begin{eqnarray*}
 &  & S_{n,k}^{*}\bar{V}_{n,k}(x,x')\\
 &  & \leq\frac{\lambda}{2}\left[V_{n,k-1}(x)+V_{n,k-1}(x')\right]+\frac{b_{d}}{\varepsilon_{d}K\left(\nu_{d},\lambda,V,b_{d}\right)}\frac{1}{2}\left[\mathbb{I}_{\mathsf{C}_{d}}(x)+\mathbb{I}_{\mathsf{C}_{d}}(x')\right]\\
 &  & \leq\frac{\lambda}{2}\left[V_{n,k-1}(x)+V_{n,k-1}(x')\right]+\left(\bar{\lambda}-\lambda\right)\left(\bar{d}+1\right)\frac{1}{2}\left[\mathbb{I}_{\mathsf{C}_{\bar{d}}}(x)+\mathbb{I}_{\mathsf{C}_{\bar{d}}}(x')\right]\\
 &  & \leq\frac{\lambda}{2}\left[V_{n,k-1}(x)+V_{n,k-1}(x')\right]+\left(\bar{\lambda}-\lambda\right)\frac{1}{2}\left[V(x)+V(x')\right]\left[\mathbb{I}_{\mathsf{C}_{\bar{d}}}(x)+\mathbb{I}_{\mathsf{C}_{\bar{d}}}(x')\right]\\
 &  & \leq\frac{\lambda}{2}\left[V_{n,k-1}(x)+V_{n,k-1}(x')\right]+\left(\bar{\lambda}-\lambda\right)\frac{1}{2}\left[V_{n,k-1}(x)+V_{n,k-1}(x')\right]\left[\mathbb{I}_{\mathsf{C}_{\bar{d}}}(x)+\mathbb{I}_{\mathsf{C}_{\bar{d}}}(x')\right]\\
 &  & \leq\bar{\lambda}\bar{V}_{n,k-1}(x,x'),\end{eqnarray*}
where for the first inequality (A\ref{hyp:norm_const_bound_below})
has been applied, the second inequality is due to the definition of
$\bar{d}$ and the penultimate inequality is due to the definition
of $\bar{V}_{n,k-1}$. For all $(x,x')\in\mathsf{\bar{C}}_{\bar{d}}$,\begin{eqnarray}
 &  & S_{n,k}^{*}\bar{V}_{n,k}(x,x')\nonumber \\
 &  & =\bar{R}_{n,k}\bar{V}_{n,k}(x,x')\nonumber \\
 &  & =\frac{1}{2(1-\bar{\epsilon}_{n,k})}\left[S_{n,k}V_{n,k}(x)+S_{n,k}V_{n,k}(x')-2\epsilon_{n,k}\bar{\nu}_{n,k}\left(V_{n,k}\right)\right]\nonumber \\
 &  & \leq\frac{\lambda}{2(1-\bar{\epsilon}_{n,k})}\left[V_{n,k-1}(x)+V_{n,k-1}(x')\right]+\frac{b_{d}}{2(1-\bar{\epsilon}_{n,k})\varepsilon_{d}\nu_{d}\left(\widetilde{Q}_{n,k:n}(1)\right)}\left[\mathbb{I}_{\mathsf{C}_{d}}(x)+\mathbb{I}_{\mathsf{C}_{d}}(x')\right]\nonumber \\
 &  & \leq\frac{\lambda\bar{d}}{(1-\bar{\epsilon}_{n,k})}\frac{1}{\inf_{x:V(x)\leq\bar{d}}M_{n,k}\left(\widetilde{Q}_{n,k:n}(1)\right)(x)}+\frac{b_{d}}{(1-\bar{\epsilon}_{n,k})\varepsilon_{d}\nu_{d}\left(\widetilde{Q}_{n,k:n}(1)\right)}=:\bar{b}_{n,k-1},\label{eq:defn_bar_b_n_p-1}\end{eqnarray}
where equation (\ref{eq:uni_var_nh_drift}) has been used. This concludes
the proof of equation (\ref{eq:bi_var_drift-1}). Applying (A\ref{hyp:norm_const_bound_below})
to the denominator terms in (\ref{eq:defn_bar_b_n_p-1}) and using
(\ref{eq:eps_n_p_bounded_below_proof1})-(\ref{eq:eps_n_p_bouned_below_proof2})
establishes the remaining part of equation (\ref{eq:eps_n_p_bounded_below}). 
\end{proof}

\begin{proof}
\emph{(Theorem \ref{thm:v_norm_bound}).} Let $d_{0}$ and $\lambda$
be as in (A\ref{hyp:fosterlyapunov}). Set $d\geq d_{0}$, $\bar{\lambda}\in(\lambda,1)$
and let $\bar{d}$, $\left\{ \bar{\epsilon}_{n,k}\right\} $, $\left\{ \bar{b}_{n,k}\right\} $,
$\left\{ \bar{R}_{n,k}\right\} $, and$\left\{ \bar{V}_{n,k}\right\} $
be as in Proposition \ref{pro:bi_var_drift}. The latter verifies
conditions (NS1) and (NS2) of \citet{mc:theory:DMR:04}. Consequently
\citep[Theorem 8]{mc:theory:DMR:04} may be applied. For $\alpha=1$,
the uniform bounds in equation (\ref{eq:eps_n_p_bounded_below}) of
Proposition \ref{pro:bi_var_drift}, combined with standard manipulations
of the bounds of \citet[Theorem 8]{mc:theory:DMR:04} (details omitted
for brevity) show that there exists a finite constant $M$ and $\rho<1$
such that

\begin{eqnarray}
\left\Vert \mu S_{n,k}\ldots S_{n,n}-\mu S_{n,k}\ldots S_{n,n}\right\Vert _{V} & \leq & M\rho^{n-k+1}\left[\mu\left(V_{n,k-1}\right)+\mu'\left(V_{n,k-1}\right)\right],\label{eq:theorem_quant_bound_proof}\end{eqnarray}
Noting that from equation (\ref{eq:transport_in_terms_of_S}),\begin{eqnarray*}
\Phi_{n,k:n}(\mu)(A) & = & \frac{\mu\left(\widetilde{G}_{n,k}M_{n,k+1}\left(\widetilde{Q}_{n,k+1:n}(1)\right)S_{n,k+1}\ldots S_{n,n}(A)\right)}{\mu\left(\widetilde{Q}_{n,k:n}(1)\right)},\end{eqnarray*}
equation (\ref{eq:theorem_v_norma_bounds_Phi_state}) holds due to
(\ref{eq:theorem_quant_bound_proof}) and the definition of $V_{n,k-1}$
given in equation (\ref{eq:V_n_p_prelim}). For the case $\alpha\in(0,1)$,
due to Jensen's inequality and the fact that for any two non-negative
reals $a,b$ and $\alpha\in[0,1]$, $(a+b)^{\alpha}\leq a^{\alpha}+b^{\alpha}$,
we have that whenever equation (\ref{eq:bi_var_drift-1}) of Proposition
\ref{pro:bi_var_drift} holds, \begin{eqnarray*}
S_{n,k}^{*}\left(\bar{V}_{n,k}^{\alpha}\right)(x,x') & \leq & \left(\bar{\lambda}\bar{V}_{n,k-1}(x,x')+\bar{b}_{n,k-1}\mathbb{I}_{\bar{\mathsf{C}}}(x,x')\right)^{\alpha}\\
 & \leq & \bar{\lambda}^{\alpha}\bar{V}_{n,k-1}^{\alpha}(x,x')+\bar{b}_{n,k-1}^{\alpha}\mathbb{I}_{\bar{\mathsf{C}}}(x,x'),\quad\forall(x,x')\in\mathsf{X}\times\mathsf{X}.\end{eqnarray*}
and for $\bar{V}_{n,k-1}$ given in equation (\ref{eq:bi_var_drift_in_terms_of_uni}),
\begin{eqnarray*}
\bar{V}_{n,k-1}^{\alpha}(x,x') & \leq & \frac{1}{2^{\alpha}}\left[V_{n,k}^{\alpha}(x)+V_{n,k}^{\alpha}(x')\right].\end{eqnarray*}
The arguments as for the case $\alpha=1$ are then repeated essentially
replacing $\bar{V}_{n,k}$, $\bar{\lambda}$, $\overline{b}_{n,k-1}$
by $\bar{V}_{n,k}^{\alpha}$, $\bar{\lambda}^{\alpha}$ and $\overline{b}_{n,k-1}^{\alpha}$
respectively, in order to establish equation (\ref{eq:theorem_v_norm_bounds_statement})
and thus (\ref{eq:theorem_v_norma_bounds_Phi_state})\@. The details
are omitted for brevity.
\end{proof}

\subsection{Verifying assumption (A\ref{hyp:norm_const_bound_below})\label{sub:Verifying-assumption-(Adeterm)}}

The following lemma illustrates that (A\ref{hyp:norm_const_bound_below})
can be verified under a generic condition on the decay in $x$ of
the potential functions $\left\{ G_{n,k}\right\} $ specified via
$\left\{ U_{n,k}\right\} $, relative to the drift function $V$ of
assumption (A\ref{hyp:fosterlyapunov}). 
\begin{lem}
\label{lem:norm_const_bounded_below} Assume (A\ref{hyp:G_bounded-1})
and (A\ref{hyp:fosterlyapunov}). Let $V$, $d_{0}$ and $\lambda$
be as in (A\ref{hyp:fosterlyapunov}) and assume\[
\sup_{n\geq1}\sup_{0\leq k\leq n-1}\sup_{x\in\mathsf{X}}\frac{U_{n,k}(x)}{V(x)}<+\infty.\]
For any $d\geq d_{0}$ let $b_{d}$ be the corresponding constant
of (A\ref{hyp:fosterlyapunov}). Then there exists a positive, finite
constant $C$ depending only on $\lambda$ and $b_{d}$ such that
for any $\mu\in\mathcal{P}(\mathsf{X})$ with $\mu(V)<+\infty$,\begin{eqnarray}
\inf_{n\geq1}\inf_{0\leq k\leq n}\mu\left(\widetilde{Q}_{n,k:n}(1)\right) & \geq & \exp\left[-C\mu(V)\right].\label{eq:norm_const_unif bound}\end{eqnarray}
\end{lem}
\begin{proof}
Firstly, $\mu\left(\widetilde{Q}_{n,n:n}(1)\right)=\mu(1)$ and by
Jensen's inequality, \begin{eqnarray*}
\mu\left(\widetilde{Q}_{n,n-1:n}(1)\right) & =\mu\left(\widetilde{G}_{n,n-1}\right) & \geq\exp\left[-\frac{1}{n}\mu(U_{n,n-1})\right]\geq\exp\left[-\mu(V)\left\Vert U_{n,n-1}\right\Vert _{V}\right].\end{eqnarray*}
For $1\leq k<n-1$, by Jensen's inequality,\begin{eqnarray}
\mu\left(\widetilde{Q}_{n,k:n}(1)\right) & = & \int_{\mathsf{X}^{n-k+1}}\exp\left(-\frac{1}{n}\sum_{\ell=k}^{n-1}U_{n,\ell}(x_{\ell})\right)\mu(dx_{k})\prod_{\ell=k+1}^{n}M_{n,\ell}(x_{\ell-1},dx_{\ell})\nonumber \\
 & \geq & \exp\left[-\frac{1}{n}\int_{\mathsf{X}^{n-k+1}}\left(\sum_{\ell=k}^{n-1}U_{n,\ell}(x_{\ell})\right)\mu(dx_{k})\prod_{\ell=k+1}^{n}M_{n,\ell}(x_{\ell-1},dx_{\ell})\right]\nonumber \\
 & = & \exp\left[-\frac{1}{n}\mu(U_{n,k})-\frac{1}{n}\sum_{\ell=k+1}^{n-1}\int_{\mathsf{X}}U_{n,\ell}(x_{\ell})\mu M_{n,k:\ell}(dx_{\ell})\right].\label{eq:norm_const_lem_bound}\end{eqnarray}
Iteration of the drift inequality in (A\ref{hyp:fosterlyapunov})
shows that for any $1\leq k<\ell<n$, \begin{eqnarray}
\int_{\mathsf{X}}V(x_{\ell})M_{n,k:\ell}(x_{k},dx_{\ell}) & \leq & \lambda^{\ell-k}V(x)+b_{d}\sum_{j=0}^{\ell-k-1}\lambda^{j}.\label{eq:lemm_norm_const_induct}\end{eqnarray}
It follows from (\ref{eq:lemm_norm_const_induct}) that\begin{eqnarray*}
\int_{\mathsf{X}}U_{n,\ell}(x_{\ell})\mu M_{n,k-1:\ell}(dx_{\ell}) & \leq & \left\Vert U_{n,\ell}\right\Vert _{V}\int_{\mathsf{X}}V(x_{\ell})\mu M_{n,k-1:\ell}(dx_{\ell})\leq\mu(V)+b_{d}\frac{1}{1-\lambda}\end{eqnarray*}
which combined with (\ref{eq:norm_const_lem_bound}) implies the desired
result. 
\end{proof}

\section{$\mathbb{L}_{p}$ error bounds for the particle measures\label{sec:-Error-Bounds}}

Making use of the results of section \ref{sec:Stability-of-model},
the following theorem presents an $\mathbb{L}_{p}$ bound on the error
$\eta_{n,n}^{N}(f)-\eta_{n,n}(f)$, for some possibly unbounded $f$.
This theorem rests on assumptions about the moments of the mean particle
drift, $\eta_{n,k}^{N}(V)$, and a related normalization quantity,
which are used in the proof to bound the moments of Martingale increments
associated with the particle approximation. Discussion of the (\ref{eq:E^N[eta(V)]_bounded})
is given in the appendix and both assumptions are verified in the
application of section \ref{sec:Application}.
\begin{thm}
\label{thm:L_p_bound}Assume (A\ref{hyp:G_bounded-1}), (A\ref{hyp:fosterlyapunov})
and (A\ref{hyp:norm_const_bound_below}). Let $V$ be as in (A\ref{hyp:fosterlyapunov})
and for $s>0$ an independent parameter let $t:=\frac{1+s}{s}$. Let
$p\geq1$ and $\alpha\in[0,1]$, be such that $\alpha tp\leq1$ and
$(1+s)p\leq1$, and for $\mu\in\mathcal{P}(\mathsf{X})$ assume \begin{eqnarray}
\sup_{N\geq1}\sup_{n\geq1}\sup_{1\leq k\leq n}\bar{\mathbb{E}}_{\mu}\left[\eta_{n,k}^{N}\left(\widetilde{Q}_{n,k:n}(1)\right)^{-(1+s)p}\right] & < & +\infty,\label{eq:E^N[eta(Q)]_bounded_below}\\
\sup_{N\geq1}\sup_{n\geq1}\sup_{1\leq k\leq n}\bar{\mathbb{E}}_{\mu}\left[\eta_{n,k}^{N}\left(V^{\alpha tp}\right)\right] & < & +\infty.\label{eq:E^N[eta(V)]_bounded}\end{eqnarray}
Then there exists $\rho<1$ and a finite constant $C$ depending on
$\alpha$, $\mu$, $V$, and the constants in (A\ref{hyp:G_bounded-1}),
(A\ref{hyp:fosterlyapunov}), and (A\ref{hyp:norm_const_bound_below})
such that for any $f\in\mathcal{L}_{V^{\alpha}}$, $n\in\mathbb{N}$
and $N\in\mathbb{N}$,\begin{eqnarray}
\bar{\mathbb{E}}_{\mu}\left[\left|\left(\eta_{n,n}^{N}-\eta_{n,n}\right)(f)\right|^{p}\right]^{1/p} & \leq & C\left\Vert f\right\Vert _{V^{\alpha}}\left(\frac{1-\rho^{n}}{1-\rho}\right)\frac{1}{\sqrt{N}}.\label{eq:_phi_mu_v_bounded}\end{eqnarray}

\end{thm}

\begin{proof}
Throughout the proof $C$ denotes a constant whose value may change
on each appearance. Consider the telescoping decomposition\begin{eqnarray}
\left(\eta_{n,n}^{N}-\eta_{n,n}\right)(f) & = & \sum_{k=0}^{n}\left[\Phi_{n,k:n}\left(\eta_{n,k}^{N}\right)-\Phi_{n,k:n}\left(\Phi_{n,k}\left(\eta_{n,k-1}^{N}\right)\right)\right](f),\label{eq:telescope_decomp}\end{eqnarray}
with the convention $\Phi_{n,0}\left(\eta_{n,-1}^{N}\right):=\mu$.
For any of the terms in the summation of equation (\ref{eq:telescope_decomp}),
following the approach of \citet[page 245]{smc:theory:Dm04}, we have

\begin{eqnarray}
 &  & \left[\Phi_{n,k:n}\left(\eta_{n,k}^{N}\right)-\Phi_{n,k:n}\left(\Phi_{n,k}\left(\eta_{n,k-1}^{N}\right)\right)\right](f)\nonumber \\
 &  & =\left[\frac{\eta_{n,k}^{N}\widetilde{Q}_{n,k:n}}{\eta_{n,k}^{N}\widetilde{Q}_{n,k:n}(1)}-\frac{\Phi_{n,k}\left(\eta_{n,k-1}^{N}\right)\widetilde{Q}_{n,k:n}}{\Phi_{n,k}\left(\eta_{n,k-1}^{N}\right)\widetilde{Q}_{n,k:n}(1)}\right](f)\nonumber \\
 &  & =\frac{1}{\eta_{n,k}^{N}\widetilde{Q}_{n,k:n}(1)}\left[\eta_{n,k}^{N}\widetilde{Q}_{n,k:n}-\eta_{n,k}^{N}\widetilde{Q}_{n,k:n}(1)\frac{\Phi_{n,k}\left(\eta_{n,k-1}^{N}\right)\widetilde{Q}_{n,k:n}}{\Phi_{n,k}\left(\eta_{n,k-1}^{N}\right)\widetilde{Q}_{n,k:n}(1)}\right](f)\nonumber \\
 &  & =\frac{1}{\eta_{n,k}^{N}\widetilde{Q}_{n,k:n}(1)}\left[\eta_{n,k}^{N}-\Phi_{n,k}\left(\eta_{n,k-1}^{N}\right)\right]\widetilde{Q}_{n,k:n}^{N}(f),\label{eq:tele_1_term_decomp}\end{eqnarray}
where\[
\widetilde{Q}_{n,k:n}^{N}(f)(x):=\widetilde{Q}_{n,k:n}(f)(x)-\widetilde{Q}_{n,k:n}(1)(x)\frac{\Phi_{n,k}\left(\eta_{n,k-1}^{N}\right)\widetilde{Q}_{n,k:n}(f)}{\Phi_{n,k}\left(\eta_{n,k-1}^{N}\right)\widetilde{Q}_{n,k:n}(1)}.\]

From equations (\ref{eq:telescope_decomp}) and (\ref{eq:tele_1_term_decomp}),
for $p\geq1$,\begin{eqnarray}
 &  & \bar{\mathbb{E}}_{\mu}\left[\left|\left(\eta_{n,n}^{N}-\eta_{n,n}\right)(f)\right|^{p}\right]^{1/p}\nonumber \\
 &  & \leq\sum_{k=0}^{n}\bar{\mathbb{E}}_{\mu}\left[\left|\frac{1}{\eta_{n,k}^{N}\widetilde{Q}_{n,k:n}(1)}\left[\eta_{n,k}^{N}-\Phi_{n,k}\left(\eta_{n,k-1}^{N}\right)\right]\widetilde{Q}_{n,k:n}^{N}(f)\right|^{p}\right]^{1/p}\nonumber \\
 &  & \leq\sum_{k=0}^{n}\bar{\mathbb{E}}_{\mu}\left[\left|\frac{1}{\eta_{n,k}^{N}\widetilde{Q}_{n,k:n}(1)}\right|^{(1+s)p}\right]^{1/\left[(1+s)p\right]}\bar{\mathbb{E}}_{\mu}\left[\left|\left[\eta_{n,k}^{N}-\Phi_{n,k}\left(\eta_{n,k-1}^{N}\right)\right]\widetilde{Q}_{n,k:n}^{N}(f)\right|^{pt}\right]^{1/\left(tp\right)},\label{eq:telescope_LP_bound_strat}\end{eqnarray}
where Minkowski's and Hölder's inequalities have been applied. We
next proceed to bound each of the factors in the summands of equation
(\ref{eq:telescope_LP_bound_strat}).

Denoting\begin{eqnarray*}
\left[\eta_{n,k}^{N}-\Phi_{n,k}\left(\eta_{n,k-1}^{N}\right)\right]\widetilde{Q}_{n,k:n}^{N}(f) & = & \frac{1}{N}\sum_{i=1}^{N}T_{n,k}^{(i)},\end{eqnarray*}
 where \begin{eqnarray*}
T_{n,k}^{(i)} & := & \widetilde{Q}_{n,k:n}^{N}(f)(\xi_{n,k}^{\left(N,i\right)})-\Phi_{n,k}\left(\eta_{n,k-1}^{N}\right)\widetilde{Q}_{n,k:n}^{N}(f),\end{eqnarray*}
(with the dependence of $T_{n,k}^{(i)}$ on $N$ suppressed) we have
that for any $n\in\mathbb{N}$, $1\leq k\leq n$ and $1\leq i\leq N$,

\[
\bar{\mathbb{E}}_{\mu}\left[\left.T_{n,k}^{(i)}\right|\mathcal{F}_{n,k}^{(N,i-1)}\right]=0,\]
with the convention that $\mathcal{F}_{n,k}^{(N,0)}=\mathcal{F}_{n,k-1}^{(N,N)}$.
Next, there exists a constant $C$ such that for any $p\geq1$,\begin{eqnarray}
 &  & \bar{\mathbb{E}}_{\mu}\left[\left|T_{n,k}^{(i)}\right|^{p}\right]^{1/p}\nonumber \\
 &  & =\bar{\mathbb{E}}_{\mu}\Biggl[\Biggl|\widetilde{Q}_{n,k:n}(1)\left(\xi_{n,k}^{\left(N,i\right)}\right)\Biggl[S_{n,k+1}\ldots S_{n,n}(f)\left(\xi_{n,k}^{\left(N,i\right)}\right)\nonumber \\
 &  & \quad\left.-\left.\left.\frac{\Phi_{n,k}\left(\eta_{n,k-1}^{N}\right)\left(\widetilde{Q}_{n,k:n}(1)S_{n,k+1}\ldots S_{n,n}(f)\right)}{\Phi_{n,k}\left(\eta_{n,k-1}^{N}\right)\left(\widetilde{Q}_{n,k:n}(1)\right)}\right]\right|^{p}\right]^{1/p}\nonumber \\
 &  & \leq\rho^{n-k}M\left\Vert f\right\Vert _{V^{\alpha}}\bar{\mathbb{E}}_{\mu}\left[\left|\widetilde{Q}_{n,k:n}(1)\left(\xi_{n,k}^{\left(N,i\right)}\right)\left[V_{n,k}^{\alpha}\left(\xi_{n,k}^{\left(N,i\right)}\right)+\frac{\Phi_{n,k}\left(\eta_{n,k-1}^{N}\right)\left(\widetilde{Q}_{n,k:n}(1)V_{n,k}^{\alpha}\right)}{\Phi_{n,k}\left(\eta_{n,k-1}^{N}\right)\left(\widetilde{Q}_{n,k:n}(1)\right)}\right]\right|^{p}\right]^{1/p}\nonumber \\
 &  & \leq\rho^{n-k}M\left\Vert f\right\Vert _{V^{\alpha}}\bar{\mathbb{E}}_{\mu}\left[V^{\alpha p}\left(\xi_{n,k}^{\left(N,i\right)}\right)\right]^{1/p}\nonumber \\
 &  & \quad+\rho^{n-k}M\left\Vert f\right\Vert _{V^{\alpha}}\bar{\mathbb{E}}_{\mu}\left[\left|\widetilde{Q}_{n,k:n}(1)\left(\xi_{n,k}^{\left(N,i\right)}\right)\frac{\Phi_{n,k}\left(\eta_{n,k-1}^{N}\right)\left(\widetilde{Q}_{n,k:n}(1)V_{n,k}^{\alpha}\right)}{\Phi_{n,k}\left(\eta_{n,k-1}^{N}\right)\left(\widetilde{Q}_{n,k:n}(1)\right)}\right|^{p}\right]^{1/p}\nonumber \\
 &  & \leq\rho^{n-k}M\left\Vert f\right\Vert _{V^{\alpha}}\bar{\mathbb{E}}_{\mu}\left[V^{\alpha p}\left(\xi_{n,k}^{\left(N,i\right)}\right)\right]^{1/p}\nonumber \\
 &  & \quad+\rho^{n-k}M\left\Vert f\right\Vert _{V^{\alpha}}\bar{\mathbb{E}}_{\mu}\left[\frac{\Phi_{n,k}\left(\eta_{n,k-1}^{N}\right)\left(\widetilde{Q}_{n,k:n}(1)V_{n,k}^{\alpha p}\right)}{\Phi_{n,k}\left(\eta_{n,k-1}^{N}\right)\left(\widetilde{Q}_{n,k:n}(1)\right)}\bar{\mathbb{E}}_{\mu}\left[\left.\widetilde{Q}_{n,k:n}(1)\left(\xi_{n,k}^{(N,i)}\right)\right|\mathcal{F}_{n,k}^{\left(N,i-1\right)}\right]\right]^{1/p}\nonumber \\
 &  & \leq\rho^{n-k}M\left\Vert f\right\Vert _{V^{\alpha}}\left(\bar{\mathbb{E}}_{\mu}\left[\eta_{n,k}^{N}\left(V^{\alpha p}\right)\right]^{1/p}+\bar{\mathbb{E}}_{\mu}\left[\Phi_{n,k}\left(\eta_{n,k-1}^{N}\right)\left(V^{\alpha p}\right)\right]^{1/p}\right)\nonumber \\
 &  & \leq\rho^{n-k}C\left\Vert f\right\Vert _{V^{\alpha}},\label{eq:L_p_bound_on_increment}\end{eqnarray}
where Theorem \ref{thm:v_norm_bound}, followed by Minkowski's inequality,
Jensen's inequality, the exchangeability property of equation (\ref{eq:exch}),
Jensen's inequality again and the assumption of equation (\ref{eq:E^N[eta(V)]_bounded})
have been applied. Thus for fixed $N$, $\left\{ \left(\sum_{j=1}^{i}T_{n,k}^{(j)},\mathcal{F}_{n,k}^{(N,i)}\right);1\leq i\leq N\right\} $
is a Martingale sequence with increments bounded in $\mathbb{L}_{p}$.
It follows that when $tp\geq2,$ by the Burkholder-Davis inequality
and Minkowski's inequality, there exists a constant $C$ such that\begin{eqnarray*}
\bar{\mathbb{E}}_{\mu}\left[\left|\left[\eta_{n,k}^{N}-\Phi_{n,k}\left(\eta_{n,k-1}^{N}\right)\right]\widetilde{Q}_{n,k:n}^{N}(f)\right|^{tp}\right]^{1/(tp)} &  & \leq CN^{-1}\bar{\mathbb{E}}_{\mu}\left[\left|\sum_{i=1}^{N}\left(T_{n,k}^{(i)}\right)^{2}\right|^{tp/2}\right]^{1/(tp)}\\
 &  & \leq CN^{-1}\left(\sum_{i=1}^{N}\bar{\mathbb{E}}_{\mu}\left[\left|T_{n,k}^{(i)}\right|^{tp}\right]^{2/(tp)}\right)^{1/2},\end{eqnarray*}
and when $1<tp<2$, using the fact that for any $a,b\geq0$ and $0\leq r\leq1$,
$(a+b)^{r}\leq a^{r}+b^{r}$, \begin{eqnarray*}
\bar{\mathbb{E}}_{\mu}\left[\left|\left[\eta_{n,k}^{N}-\Phi_{n,k}\left(\eta_{n,k-1}^{N}\right)\right]\widetilde{Q}_{n,k:n}^{N}(f)\right|^{tp}\right]^{1/(tp)} &  & \leq CN^{-1}\bar{\mathbb{E}}_{\mu}\left[\left|\sum_{i=1}^{N}\left(T_{n,k}^{(i)}\right)^{2}\right|^{tp/2}\right]^{1/(tp)}\\
 &  & \leq CN^{-1}\left(\sum_{i=1}^{N}\bar{\mathbb{E}}_{\mu}\left[\left(T_{n,k}^{(i)}\right)^{tp}\right]\right)^{1/(tp)}.\end{eqnarray*}
Combining with (\ref{eq:L_p_bound_on_increment}), we conclude that
there exists a constant $C$ such that

\begin{eqnarray*}
\bar{\mathbb{E}}_{\mu}\left[\left|\left[\eta_{n,k}^{N}-\Phi_{n,k}\left(\eta_{n,k-1}^{N}\right)\right]\widetilde{Q}_{n,k:n}^{N}(f)\right|^{tp}\right]^{1/\left(tp\right)} & \leq & \rho^{n-k}C\left\Vert f\right\Vert _{V^{\alpha}}N^{^{-\left\{ tp/2\wedge\left(tp-1\right)\right\} /\left(tp\right)}},\end{eqnarray*}
for all $n\geq1$ and $1\leq k\leq n$.

The remaining terms in (\ref{eq:telescope_LP_bound_strat}) are treated
directly by the assumption of equation (\ref{eq:E^N[eta(Q)]_bounded_below}),
and therefore upon returning to (\ref{eq:telescope_LP_bound_strat})
we conclude that there exists a constant $C$ such that

\begin{eqnarray*}
\bar{\mathbb{E}}_{\mu}\left[\left|\left(\eta_{n,n}^{N}-\eta_{n,n}\right)(f)\right|^{p}\right]^{1/p} & \leq & C\left\Vert f\right\Vert _{V^{\alpha}}\frac{1}{\sqrt{N}}\sum_{k=0}^{n}\rho^{n-k}<C\left\Vert f\right\Vert _{V^{\alpha}}\left(\frac{1-\rho^{n}}{1-\rho}\right)\frac{1}{\sqrt{N}},\end{eqnarray*}
and the result holds.
\end{proof}

\section{Application \label{sec:Application}}

In this section is concerned with the case in which $\mathsf{X}=\mathbb{R}^{d}$,
$\mathcal{B}(\mathbb{R}^{d})$ is the corresponding Borel $\sigma$-algebra
and throughout consider the following structural definitions and assumptions.
\begin{itemize}
\item Let $\pi\in\mathcal{P}(\mathsf{X})$ be a target distribution admitting
a density with respect to Lebesgue measure. Also denote by $\pi$
its density. In applications of interest, this density will only be
known up to a multiplicative constant, $Z$, and denote by $\bar{\pi}$
the unnormalised density, i.e. $\pi(x)=\bar{\pi}(x)/Z,\;\; x\in\mathbb{R}^{d}$. 
\item For $\underline{\gamma}\in(0,1]$ a constant, let $\gamma:[0,1]\rightarrow[\underline{\gamma},1]$
be a non-decreasing, Lipschitz function. 
\item Let$\left\{ \pi_{\gamma};\gamma\in[\underline{\gamma},1]\right\} $
be the family of probability measures defined by\begin{eqnarray*}
\pi_{\gamma}(A) & := & \frac{\int_{A}\bar{\pi}^{\gamma}(x)\mathrm{d}x}{\int_{\mathsf{X}}\bar{\pi}^{\gamma}(x)\mathrm{d}x},\quad A\in\mathcal{B}(\mathsf{X}).\end{eqnarray*}

\item Let $q\in\mathcal{P}(\mathsf{X})$ be an increment distribution admitting
a density with respect to Lebesgue measure, also denoted by $q$.
For each $n\geq1$ and $1\leq k\leq n$, let $M_{n,k}$ be a random
walk Metropolis (RWM) kernel of invariant distribution $\pi_{\gamma(k/n)}$
and proposal kernel $q$ , i.e.\begin{eqnarray*}
M_{n,k}(x,A) & = & \int_{A-x}\left[1\wedge\dfrac{\bar{\pi}^{\gamma(k/n)}(x+y)}{\bar{\pi}^{\gamma(k/n)}(x)}\right]q(y)\lambda^{\text{Leb}}(dy)\\
 &  & +\delta_{x}(A)\int_{\mathsf{X}-x}\left[1-\left[1\wedge\dfrac{\bar{\pi}^{\gamma(k/n)}(x+y)}{\bar{\pi}^{\gamma(k/n)}(x)}\right]\right]q(y)\lambda^{\text{Leb}}(dy),\quad A\in\mathcal{B}(\mathsf{X}).\end{eqnarray*}
where for any set $C$, $C-x:=\left\{ z\in\mathsf{X};\; z+x\in C\right\} $
(note that in applications it may be of interest to allow $q$ to
depend on $n$ and $k$, for example via $\gamma(k/n)$, but for simplicity
this issue is not pursued further here).
\item Let $\left\{ G_{n,k};n\geq1,0\leq k\leq n-1\right\} $ be a collection
of potential functions defined by\begin{eqnarray*}
G_{n,k}(x) & = & \exp\left[\frac{1}{n}\log\bar{\pi}(x)\left(\frac{\gamma((k+1)/n)-\gamma(k/n)}{1/n}\right)\right].\end{eqnarray*}

\end{itemize}
Consider the following assumptions on the target density$\pi$ and
increment density $q$.
\begin{itemize}
\item The density $\pi$ is strictly positive, bounded and has continuous
first derivatives such that\begin{eqnarray}
\lim_{r\rightarrow\infty}\sup_{\left|x\right|\geq r}n(x)\cdot\nabla\log\pi(x)=-\infty, &  & \lim_{r\rightarrow\infty}\sup_{\left|x\right|\geq r}n(x)\cdot\dfrac{\nabla\pi(x)}{\left|\nabla\pi(x)\right|}<0,\label{eq:pi_super_exponential}\end{eqnarray}

\item For all $r>0$ there exists $\epsilon_{r}>0$ such that \begin{eqnarray}
\left|x\right|\leq r & \Rightarrow & q(x)\geq\epsilon_{r}.\label{eq:increment density}\end{eqnarray}

\end{itemize}
The assumptions of equations (\ref{eq:pi_super_exponential})-(\ref{eq:increment density})
are standard types of assumptions ensuring geometric ergodicity of
RWM kernels \citep{mc:theory:rt96,mc:theory:JH00}. The assumption
of equation (\ref{eq:increment density}) is stronger than the standard
one in \citep{mc:theory:JH00}, but is flexible enough to verify (A\ref{hyp:fosterlyapunov})
which involves a family of minorization measures/constants, indexed
over a range of levels of $V$. 

The interest in the specific FK models of this section arises from
the choice of the initial distribution $\mu$ addressed in the following
lemma. This FK model corresponds to a particular choice of the {}``backwards''
kernels in \citep{smc:meth:DDJ06}, and in the corresponding SMC algorithm
the order of the weighting and resampling steps is reversed. 
\begin{lem}
\label{lem:smc_sampler_unbiased}Consider the operators \textup{$\left\{ \Phi_{n,k}\right\} $}
associated with $\left\{ G_{n,k}\right\} $ and $\left\{ M_{n,k}\right\} $
of section \ref{sec:Application}. Then for all $n\geq1$ and $0\leq k\leq n$,
$\Phi_{n,0:k}\left(\pi_{\underline{\gamma}}\right)=\pi_{\gamma(k/n)}$. \end{lem}
\begin{proof}
Fix $n$ arbitrarily and suppose the result holds at rank $0<k<n$.
Then

\begin{eqnarray*}
\Phi_{n,0:k+1}\left(\pi_{\underline{\gamma}}\right)(A) & = & \Phi_{n,k+1}\circ\Phi_{n,0:k}\left(\pi_{\underline{\gamma}}\right)(A)\\
 & = & \frac{\int\dfrac{\bar{\pi}^{\gamma(\left(k+1\right)/n)}(x)}{\bar{\pi}^{\gamma(k/n)}(x)}M_{n,k+1}(A)(x)\pi_{\gamma(k/n)}(dx)}{\int\dfrac{\bar{\pi}^{\gamma(\left(k+1\right)/n)}(x)}{\bar{\pi}^{\gamma(k/n)}(x)}\pi_{\gamma(k/n)}(dx)}\\
 & = & \pi_{\gamma(\left(k+1\right)/n)}(A),\quad A\in\mathcal{B}(\mathsf{X}),\end{eqnarray*}
due to the property that $M_{n,k}$ is invariant for $\pi_{\gamma(k/n)}(dx)$.
The proof is complete upon noting that for all $n$, $\Phi_{n,0:0}=Id$
by convention.
\end{proof}
We have the following result.
\begin{thm}
\label{thm:app}Consider the collection of FK models specified in
section \ref{sec:Application}. Let $s>0$ be an independent parameter
and set $t=\frac{1+s}{s}$. Let $\alpha\in[0,1]$, and $p\geq1$ be
such that $\alpha pt\leq1$ and $(1+s)p(1-\underline{\gamma})/\underline{\gamma}<1$.
Then there exist finite constants $C_{1}(p,\mu)$, and $C_{2}\left(\mu,\pi_{\underline{\gamma}}\right)$
(depending implicitly on $\pi$ and $\gamma(\cdot)$), and constants
$\beta\in(0,1)$ and $\rho\in[0,1),$ such that for any $f\in\mathcal{L}_{V^{\alpha}}$,
$n\geq1$ and $N\geq1$,\begin{eqnarray*}
\bar{\mathbb{E}}_{\mu}\left[\left|\left(\pi_{n}^{N}-\pi\right)(f)\right|^{p}\right]^{1/p} & \leq & \left\Vert f\right\Vert _{V^{\alpha}}\left(\frac{C_{1}(p,\mu)}{\sqrt{N}}+\rho^{n}C_{2}\left(\mu,\pi_{\underline{\gamma}}\right)\mathbb{I}\left[\mu\neq\pi_{\underline{\gamma}}\right]\right),\end{eqnarray*}
 where $V(x)\propto\pi^{-\beta\underline{\gamma}}(x)$ and for each
$n$, $\pi_{n}^{N}:=\eta_{n,n}^{N}$.
\end{thm}
The proof of Theorem \ref{thm:app} is postponed until after the following
proposition regarding the verification of assumptions. 
\begin{prop}
\label{pro:app_verify_assump}Consider the setting of section \ref{sec:Application}.
Then (A\ref{hyp:G_bounded-1}), (A\ref{hyp:fosterlyapunov}) and (A\ref{hyp:norm_const_bound_below})
hold.\end{prop}
\begin{proof}
As the density $\pi$ is bounded and $\gamma(\cdot)$ is Lipschitz,
(A\ref{hyp:G_bounded-1}) holds by the mean value theorem.  We now
turn to the verification of (A\ref{hyp:fosterlyapunov}). Various
arguments are adopted from \citet{mcmc:the:ABD01} and the manipulations
are fairly standard, but are included here for completeness. The main
difference is that we need to explicitly verify the drift and minorization
conditions of (A\ref{hyp:fosterlyapunov}) which hold over a range
of sub-levels for $V$ and the proof below involves verification of
some assumptions taken as given in \citep[Lemma 4]{mcmc:the:ABD01}.

Firstly, due to the definition of $\pi_{\underline{\gamma}}$ we have
that $\nabla\log\pi_{\underline{\gamma}}(x)=\underline{\gamma}\nabla\log\bar{\pi}(x)=\underline{\gamma}\nabla\log\pi(x)$
and $\dfrac{\nabla\pi(x)}{\left|\nabla\pi(x)\right|}=\dfrac{\bar{\pi}(x)\nabla\log\bar{\pi}(x)}{\left|\bar{\pi}(x)\nabla\log\bar{\pi}(x)\right|}\dfrac{\underline{\gamma}}{\underline{\gamma}}\dfrac{\bar{\pi}^{\underline{\gamma}-1}(x)}{\bar{\pi}^{\underline{\gamma}-1}(x)}=\dfrac{\nabla\pi_{\underline{\gamma}}(x)}{\left|\nabla\pi_{\underline{\gamma}}(x)\right|}$
and so\begin{eqnarray}
\lim_{r\rightarrow\infty}\sup_{\left|x\right|\geq r}n(x)\cdot\nabla\log\pi_{\underline{\gamma}}(x)=-\infty, &  & \lim_{r\rightarrow\infty}\sup_{\left|x\right|\geq r}n(x)\cdot\dfrac{\nabla\pi_{\underline{\gamma}}(x)}{\left|\nabla\pi_{\underline{\gamma}}(x)\right|}<0\label{eq:pi_gam_super_exponential-1}\end{eqnarray}

In order to verify (A\ref{hyp:fosterlyapunov}) we first verify a
drift condition for $M_{0}(x,dy)$, defined to be the RWM kernel reversible
w.r.t. $\pi_{\underline{\gamma}}(x)$ with increment density $q$
as above, i.e.\begin{eqnarray*}
M_{0}(x,A) & = & \int_{A-x}\left[1\wedge\dfrac{\bar{\pi}^{\underline{\gamma}}(x+y)}{\bar{\pi}^{\underline{\gamma}}(x)}\right]q(y)\lambda^{\text{Leb}}(dy)\\
 &  & +\delta_{x}(A)\int_{\mathsf{X}-x}\left[1-\left[1\wedge\dfrac{\bar{\pi}^{\underline{\gamma}}(x+y)}{\bar{\pi}^{\underline{\gamma}}(x)}\right]\right]q(y)\lambda^{\text{Leb}}(dy),\quad A\in\mathcal{B}(\mathsf{X}).\end{eqnarray*}
Let $\beta\in(0,1)$ and define $V:\mathsf{X}\rightarrow[1,+\infty)$
by\begin{eqnarray}
V(x) & := & \dfrac{\pi^{-\underline{\gamma}\beta}(x)}{\inf_{x}\pi^{-\underline{\gamma}\beta}(x)}.\label{eq:app_defn_of_V}\end{eqnarray}
The results of \citet{mc:theory:JH00} show that when (\ref{eq:pi_gam_super_exponential-1})
holds, then for $M_{0}$ with increment density $q$ satisfying equation
(\ref{eq:increment density}), it holds that $\lim_{r\rightarrow\infty}\sup_{\left|x\right|\geq r}\dfrac{M_{0}V(x)}{V(x)}<1$.
Thus there exist $\lambda<1$ and $\rho_{\lambda}<+\infty$ such that\begin{equation}
\left|x\right|\geq\rho_{\lambda}\Rightarrow\dfrac{M_{0}V(x)}{V(x)}\leq\lambda.\label{eq:MV_V_lambda}\end{equation}
Due to (\ref{eq:pi_super_exponential}), there exists $\epsilon>0$
and $\rho_{\epsilon}>0$ such that \begin{eqnarray*}
\left|x\right|\geq\rho_{\epsilon}\Rightarrow n(x)\cdot\nabla\log\pi_{\underline{\gamma}}(x) & \leq & -\epsilon.\end{eqnarray*}
Now set $r_{0}=\rho_{\lambda}\vee\rho_{\epsilon}$ and $d_{0}:=\sup_{\left|x\right|\leq r_{0}}V(x)$.
Note that $d_{0}<+\infty$ due to the definition of $V$ and as the
density $\pi$ is continuous and positive. We now proceed to verify
the drift part of (A\ref{hyp:fosterlyapunov}).

For any $d\geq d_{0}$ let $\mathsf{C}_{d}:=\left\{ x:V(x)\leq d\right\} $.
We then have \begin{eqnarray}
\sup_{x\in\mathsf{C}_{d}}M_{0}V(x) & \leq d\sup_{x\in\mathsf{C}_{d}}\dfrac{M_{0}V(x)}{V(x)} & =d\sup_{x\in\mathsf{C}_{d}}\left\{ \int_{\mathsf{A}(x)}\frac{\bar{\pi}^{-\underline{\gamma}\beta}(x+y)}{\bar{\pi}^{-\underline{\gamma}\beta}(x)}q(y)\lambda^{\text{Leb}}(dy)\right.\nonumber \\
 &  & \left.\;\;+\int_{\mathsf{R}(x)}\left[1-\frac{\bar{\pi}^{\underline{\gamma}}(x+y)}{\bar{\pi}^{\underline{\gamma}}(x)}+\frac{\bar{\pi}^{\underline{\gamma}\left(1-\beta\right)}(x+y)}{\bar{\pi}^{\underline{\gamma}\left(1-\beta\right)}(x)}\right]q(y)\lambda^{\text{Leb}}(dy)\right\} \label{eq:MV_bound_R^d}\end{eqnarray}
where $\mathsf{A}(x):=\left\{ y\in\mathsf{X}:\pi(x+y)\geq\pi(x)\right\} $,
$\mathsf{R}(x):=\left\{ y\in\mathsf{X}:\pi(x+y)<\pi(x)\right\} $.
As each of the ratios in (\ref{eq:MV_bound_R^d}) is less than or
equal to $1$, then we conclude that there exists a constant $C_{b}<+\infty$
such that \begin{eqnarray}
\sup_{x\in\mathsf{C}_{d}}M_{0}V(x) & \leq & dC_{b}=:b_{d}.\label{eq:MV_bound_R^d_1}\end{eqnarray}
Noting the definition of $d_{0}$, and combining (\ref{eq:MV_V_lambda})
and (\ref{eq:MV_bound_R^d_1}) we obtain\begin{eqnarray}
M_{0}V(x) & = & M_{0}V(x)\mathbb{I}[\left|x\right|>r_{0}]+M_{0}V(x)\mathbb{I}[\left|x\right|\leq r_{0}]\nonumber \\
 & \leq & \lambda V(x)+M_{0}V(x)\mathbb{I}[V(x)\leq d]\nonumber \\
 & \leq & \lambda V(x)+b_{d}\mathbb{I}_{\mathsf{C}_{d}}(x),\label{eq:M_0_drift}\end{eqnarray}
for any $d\geq d_{0}$ and $x\in\mathsf{X}$. The arguments of \citet[Lemma 5]{mcmc:the:ABD01}
then give $M_{n,k}V(x)\leq M_{0}V(x)$ and from this, (\ref{eq:M_0_drift})
and (\ref{eq:MV_bound_R^d_1}), we obtain for any $d\geq d_{0}$,
\begin{eqnarray*}
\sup_{n\geq1}\sup_{1\leq k\leq n}M_{n,k}V(x) & \leq & \lambda V(x)+b_{d}\mathbb{I}_{\mathsf{C}_{d}}(x),\end{eqnarray*}
which establishes the drift part of (A\ref{hyp:fosterlyapunov}).
It remains to show the minorization part. To this end we first show
that for any $d\geq d_{0}$, $\mathsf{C}_{d}$ is bounded. Recalling
the definition of $r_{0}$, we have that for any $x$ such that $\left|x\right|-r_{0}\geq0$,\begin{eqnarray*}
\frac{V(x)}{V(n(x)r_{0})} & = & \left(\frac{\pi(x)}{\pi(n(x)r_{0})}\right)^{-\underline{\gamma}\beta}\\
 & = & \exp\left[-\underline{\gamma}\beta\left(\left|x\right|-r_{0}\right)\int_{0}^{1}n(x)\cdot\nabla\log\pi(tx+(1-t)n(x)r_{0})dt\right]\\
 & \geq & \exp\left[\underline{\gamma}\beta\left(\left|x\right|-r_{0}\right)\epsilon\right]\end{eqnarray*}
from which we see that $\lim_{r\rightarrow\infty}\inf_{\left|x\right|\geq r}V(x)=+\infty$,
which in turn implies that for all $d\geq d_{0}$ there exists $r_{d}\geq0$
such that $V(x)\leq d\Rightarrow\left|x\right|\leq r_{d}.$ Then for
any $n\geq1$, $0\leq k\leq n$ and $r_{d}\geq0$, whenever $x\in\mathsf{C}_{d}$

\begin{eqnarray*}
M_{n,k}(x,A) & \geq & \int_{A-x}\left[1\wedge\frac{\pi(x+y)}{\pi(x)}\right]q(y)\lambda^{\text{Leb }}(dy)\\
 & \geq & \int_{\left(A\cap B(0,r_{d})-x\right)}\left[1\wedge\frac{\pi(x+y)}{\pi(x)}\right]q(y)\lambda^{\text{Leb }}(dy)\\
 & \geq & \epsilon_{2r_{d}}\frac{\inf_{y\in B(0,3r_{d})}\pi(y)}{\sup_{y\in B(0,3r_{d})}\pi(y)}\int_{\left(A\cap B(0,r_{d})-x\right)}\lambda^{\text{Leb }}(dy)\\
 & = & \epsilon_{2r_{d}}\frac{\inf_{y\in B(0,3r_{d})}\pi(y)}{\sup_{y\in B(0,3r_{d})}\pi(y)}\int_{B(0,r_{d})}\lambda^{\text{Leb }}(dy)\dfrac{\int_{A\cap B(0,r_{d})}\lambda^{\text{Leb }}(dy)}{\int_{B(0,r_{d})}\lambda^{\text{Leb }}(dy)}\\
 & =: & \varepsilon_{d}\cdot\nu_{d}(A),\end{eqnarray*}
where the third inequality holds due to the properties of $q$ in
(\ref{eq:increment density}) and because $A\cap B(0,r_{d})\cap B(x,2r_{d})=A\cap B(0,r_{d})$
whenever $x\in B(0,r_{d})$. As $\pi$ is strictly positive and continuous,
$V$ is bounded on compact sets and therefore $\nu_{d}(V)<+\infty$.
Also, $\mathsf{C}_{d}\supseteq\mathsf{C}_{d_{0}}$ and then due to
the definition of $d_{0}$, $\nu_{d}(\mathsf{C}_{d})\geq\nu_{d}(\mathsf{C}_{d_{0}})\geq\nu_{d}(B(0,r_{0}))>0$.
This concludes the verification of (A\ref{hyp:fosterlyapunov}).

For (A\ref{hyp:norm_const_bound_below}), from the definition of $G_{n,k}$,
we observe that $U_{n,k}$ is defined by \begin{eqnarray*}
U_{n,k}(x) & = & \log\bar{\pi}(x)\left(\frac{\gamma(\left(k+1\right)/n)-\gamma(k/n)}{1/n}\right)-C_{\gamma}\sup_{y}\log\bar{\pi}(y)\end{eqnarray*}
where $C_{\gamma}$ is the Lipschitz constant for $\gamma(\cdot)$
and we observe that $\sup_{n\geq1}\sup_{0\leq k\leq n-1}\sup_{x\in\mathsf{X}}\dfrac{U_{n,k}(x)}{V(x)}<+\infty$
. Assumption (A\ref{hyp:norm_const_bound_below}) is then satisfied
upon application of Lemma \ref{lem:norm_const_bounded_below}. 
\end{proof}

\begin{proof}
\emph{(Theorem \ref{thm:app})} 

Throughout the proof, we denote by $C$ a constant whose value may
change upon each appearance. Consider the error decomposition\begin{eqnarray}
\bar{\mathbb{E}}_{\mu}\left[\left|\left(\pi_{n}^{N}-\pi\right)(f)\right|^{p}\right]^{1/p} & = & \bar{\mathbb{E}}_{\mu}\left[\left|\left(\eta_{n,n}^{N}-\eta_{n,n}\right)(f)\right|^{p}\right]^{1/p}+\left|\left(\eta_{n,n}-\pi\right)(f)\right|.\label{eq:app_err_decomp}\end{eqnarray}

Choose $\beta$ such that $(1+s)p(1-\underline{\gamma})/\left(\underline{\gamma}\beta\right)\leq1$
and take $V$ to be defined as in equation (\ref{eq:app_defn_of_V}).
By proposition \ref{pro:app_verify_assump}, the FK model of section
\ref{sec:Application} satisfies assumptions (A\ref{hyp:G_bounded-1}),
(A\ref{hyp:fosterlyapunov}) and (A\ref{hyp:norm_const_bound_below}).
The second term on the r.h.s. of (\ref{eq:app_err_decomp}) is treated
by application of Theorem \ref{thm:v_norm_bound}. Noting that by
definition, $\eta_{n,n}=\Phi_{n,0:n}\left(\mu\right)$ and by Lemma
\ref{lem:smc_sampler_unbiased}, $\pi=\Phi_{n,0:n}\left(\pi_{\underline{\gamma}}\right)$
we obtain from Theorem \ref{thm:v_norm_bound} that there exist constants
$\rho$, $M$ and $C_{2}$ such that\begin{eqnarray*}
\left\Vert \eta_{n,n}-\pi\right\Vert _{V^{\alpha}} & \leq & M\rho^{n}\left[\frac{\mu\left(\widetilde{G}_{n,k}V^{\alpha}\right)}{\mu\left(\widetilde{Q}_{n,k:n}(1)\right)}+\frac{\pi_{\underline{\gamma}}\left(\widetilde{G}_{n,k}V^{\alpha}\right)}{\pi_{\underline{\gamma}}\left(\widetilde{Q}_{n,k:n}(1)\right)}\right]\mathbb{I}\left[\mu\neq\pi_{\underline{\gamma}}\right]\\
 & \leq & \rho^{n}C_{2}\left(\mu,\pi_{\underline{\gamma}}\right)\mathbb{I}\left[\mu\neq\pi_{\underline{\gamma}}\right].\end{eqnarray*}
where the constant $C_{2}$ arises from assumption (A\ref{hyp:norm_const_bound_below})
applied to the denominator terms and implicitly depends on $\pi_{\underline{\gamma}}$,
$V$, and the constants in (A\ref{hyp:fosterlyapunov}). 

In order to apply Theorem \ref{thm:L_p_bound} to the first term on
the r.h.s. of (\ref{eq:app_err_decomp}) it remains to verify the
assumptions of equations (\ref{eq:E^N[eta(Q)]_bounded_below})-(\ref{eq:E^N[eta(V)]_bounded}).
We start by addressing the latter. From the definition of $V$ and
due to the assumption that $\gamma(\cdot)$ is non-decreasing we observe
that for all $x,x'\in\mathsf{X}$, \begin{eqnarray*}
\left[G_{n,k}(x)-G_{n,k}(x')\right]\left[V(x)-V(x')\right] & \leq & 0\end{eqnarray*}
and therefore by Lemma \ref{lem:eta_f_g} for all (possibly random)
$\eta\in\mathcal{P}\left(\mathsf{X}\right)$, $\eta\left(G_{n,k}V\right)\leq\eta\left(G_{n,k}\right)\eta\left(V\right)$.
Then for any $n\geq1$ and $1\leq k\leq n$.\begin{eqnarray}
\bar{\mathbb{E}}_{\mu}\left[\eta_{n,k}^{N}(V)\left|\mathcal{F}_{n,k-1}^{(N,N)}\right.\right]\leq & \lambda\dfrac{\eta_{n,k-1}^{N}\left(G_{n,k-1}V\right)}{\eta_{n,k-1}^{N}\left(G_{n,k-1}\right)}+b_{d_{0}} & \leq\lambda\eta_{n,k-1}^{N}\left(V\right)+b_{d_{0}},\label{eq:app_particle_drift}\end{eqnarray}
where (A\ref{hyp:fosterlyapunov}) has been applied with $d_{0}$
is defined below equation (\ref{eq:MV_V_lambda}) in the proof of
proposition \ref{pro:app_verify_assump}. Standard iteration of the
particle drift inequality (\ref{eq:app_particle_drift}) (details
omitted for brevity) combined with the fact that $\left\{ \xi_{n,0}^{(N,i)};i=1,...,N\right\} $
are are independent and each distributed according to $\mu$ shows
that\begin{eqnarray}
\sup_{N\geq1}\sup_{n\geq1}\sup_{1\leq k\leq n}\bar{\mathbb{E}}_{\mu}\left[\eta_{n,k}^{N}\left(V\right)\right] & < & +\infty,\label{eq:app_verified_E_eta(V)_bounded}\end{eqnarray}
 and noting that $\alpha pt\leq1$, equation (\ref{eq:E^N[eta(V)]_bounded})
then holds by two applications of Jensen's inequality .

We now turn to the verification of equation (\ref{eq:E^N[eta(Q)]_bounded_below}).
From previous considerations we notice that for some finite constant
$C$,

\[
U_{n,k}(x)=\frac{1}{\beta\underline{\gamma}}\left(\frac{\gamma(\left(k+1\right)/n)-\gamma(k/n)}{1/n}\right)\log V(x)+C\]
 and therefore for any 

\begin{eqnarray*}
 &  & \eta_{n,k}^{N}\left(\widetilde{Q}_{n,k:n}(1)\right)^{-1}\\
 &  & =\left[\int\exp\left(-\frac{1}{n}\sum_{j=k}^{n-1}U_{n,j}(x_{j})\right)\eta_{n,k}^{N}(dx_{k})\prod_{j=k+1}^{n}M_{n,j}(x_{j-1,}dx_{j})\right]^{-1}\\
 &  & \leq\exp\left(\frac{1}{n}\sum_{j=k}^{n-1}\int U_{n,j}(x_{j})\eta_{n,k}^{N}M_{n,k:j}(dx_{j})\right)\\
 &  & =\exp\left(C+\frac{1}{n}\frac{1}{\beta\underline{\gamma}}\sum_{j=k}^{n-1}\left(\frac{\gamma(\left(j+1\right)/n)-\gamma(j/n)}{1/n}\right)\int\log V(x_{j})\eta_{n,k}^{N}M_{n,k:j}(dx_{j})\right)\\
 &  & \leq\exp\left(C+\frac{1}{n}\frac{1}{\beta\underline{\gamma}}\sum_{j=k}^{n-1}\left(\frac{\gamma(\left(j+1\right)/n)-\gamma(j/n)}{1/n}\right)\log\left[\int V(x_{j})\eta_{n,k}^{N}M_{n,k:j}(dx_{j})\right]\right)\\
 &  & \leq\exp\left(C+\frac{1-\gamma(k/n)}{\beta\underline{\gamma}}\log\left[\eta_{n,k}^{N}\left(V\right)\right]\right)\\
 &  & \leq\exp(C)\left[\eta_{n,k}^{N}\left(V\right)\right]^{\frac{1-\underline{\gamma}}{\beta\underline{\gamma}}}\end{eqnarray*}
where the penultimate inequality hold due to standard iteration of
the drift inequality in (A\ref{hyp:fosterlyapunov}). Therefore\begin{eqnarray*}
\bar{\mathbb{E}}_{\mu}\left[\eta_{n,k}^{N}\left(\widetilde{Q}_{n,k:n}(1)\right)^{-(1+s)p}\right] & \leq & C\bar{\mathbb{E}}_{\mu}\left[\eta_{n,k}^{N}\left(V\right)\right]^{q}\end{eqnarray*}
due to Jensen's inequality and where $q:=\frac{\left(1-\underline{\gamma}\right)}{\beta\underline{\gamma}}(1+s)p\leq1$
by assumption of the theorem. Equation (\ref{eq:E^N[eta(Q)]_bounded_below})
then follows upon combining this with equation (\ref{eq:app_verified_E_eta(V)_bounded}).
This completes the proof.
\end{proof}

\section*{Acknowledgments}

The author thanks Christophe Andrieu for discussions which lead to
the consideration of this work.

\section{Appendix}

When (A\ref{hyp:fosterlyapunov}) or even more simply (A\ref{hyp:fosterlyapunov_prelim})
holds, it is natural to ask under what further conditions, if any,
does the non-homogeneous Markov chain $\left\{ \zeta_{n,k}^{(N)};0\leq k\leq n\right\} $
also satisfy a geometric drift condition, as this would be one natural
route to verifying (\ref{eq:E^N[eta(V)]_bounded}). Defining $V_{N}:\mathsf{X}^{N}\rightarrow[1,\infty)$
by

\begin{eqnarray*}
V_{N}(\zeta) & := & \frac{1}{N}\sum_{i=1}^{N}V\left(\xi^{i}\right),\end{eqnarray*}
where $V$ is as in (A\ref{hyp:fosterlyapunov_prelim}) and $\zeta=\left(\xi^{1},...,\xi^{N}\right)$,
we see that

\begin{eqnarray}
\bar{\mathbb{E}}_{\mu}\left[\left.V_{N}(\zeta_{n,k}^{N})\right|\mathcal{F}_{n,k-1}^{\left(N,N\right)}\right] & \leq & \lambda\frac{\eta_{n,k-1}^{N}\left(\widetilde{G}_{n,k-1}V\right)}{\eta_{n,k-1}^{N}\left(\widetilde{G}_{n,k-1}\right)}+b\frac{\eta_{n,k-1}^{N}\left(\widetilde{G}_{n,k-1}\mathbb{I}_{\mathsf{C}}\right)}{\eta_{n,k-1}^{N}\left(\widetilde{G}_{n,k-1}\right)}\nonumber \\
 & \leq & \lambda\frac{\eta_{n,k-1}^{N}\left(\widetilde{G}_{n,k-1}V\right)}{\eta_{n,k-1}^{N}\left(\widetilde{G}_{n,k-1}\right)}+b\mathbb{I}_{\mathsf{C}^{(N)}}(\zeta_{n,k}^{\left(N\right)}),\label{eq:particle_chain_drift}\end{eqnarray}
where\[
\mathsf{C}^{(N)}:=\left\{ \zeta=\left(\xi^{1},...,\xi^{N}\right)\in\mathsf{X}^{N}:\exists j\in\{1,...,N\};\xi^{j}\in\mathsf{C}\right\} .\]
We are then faced with the issue of whether the re-weighting of $\eta_{n,k-1}^{N}$
by the potential function destroys the geometric drift of $M_{n,k}$:
for example one may ask when is it true that for some fixed $\delta\geq0$,
\begin{eqnarray}
\frac{\eta_{n,k-1}^{N}\left(\widetilde{G}_{n,k-1}V\right)}{\eta_{n,k-1}^{N}\left(\widetilde{G}_{n,k-1}\right)} & \leq & \left(1+\delta\right)\eta_{n,k-1}^{N}\left(V\right)\label{eq:G_destroys_drift_question}\end{eqnarray}
The remainder of this section examines this question and goes on to
look at some particular issues when $\mathsf{X}=\mathbb{R}^{d}$.
First, we have the following Lemma, which addresses a general scenario.
\begin{lem}
\label{lem:eta_f_g}For $f:\mathsf{X}\rightarrow(0,\infty)$, $g:\mathsf{X}\rightarrow(0,\infty)$,
two measurable functions and $\delta\in[0,\infty)$,

\[
\eta\left(fg\right)\leq\left(1+\delta\right)\eta(f)\eta\left(g\right)\]
for any $\eta\in\mathcal{P}(\mathsf{X})$ such that $\left|\eta\left(fg\right)\right|<+\infty,\left|\eta\left(f\right)\right|<+\infty,\left|\eta\left(g\right)\right|<+\infty$,
if

\begin{eqnarray}
\frac{\left[f(x)-f(x')\right]\left[g(x)-g(x')\right]}{\left[f(x)+f(x')\right]\left[g(x)+g(x')\right]} & \leq & \frac{\delta}{2+\delta},\quad\forall(x,x')\in\mathsf{X}^{2},\label{eq:charac_boltzman_gibbs_descend drift}\end{eqnarray}
and only if\begin{eqnarray*}
\frac{\left[f(x)-f(x')\right]\left[g(x)-g(x')\right]}{\left[f(x)+f(x')\right]\left[g(x)+g(x')\right]} & \leq & \frac{3\delta}{2+\delta},\quad\forall(x,x')\in\mathsf{X}^{2}.\end{eqnarray*}
\end{lem}
\begin{proof}
For any $\eta$, $f$, $g$ as specified in the statement and $\epsilon\in(0,1),$
consider the identity:\begin{eqnarray*}
 &  & (1-\epsilon)\eta\left(fg\right)-(1+\epsilon)\eta(f)\eta\left(g\right)\\
 &  & =\frac{1}{2}\int_{\mathsf{X}}\int_{\mathsf{X}}\left(\left[f(x)-f(x')\right]\left[g(x)-g(x')\right]-\epsilon\left[f(x)+f(x')\right]\left[g(x)+g(x')\right]\right)\eta(dx)\eta(dx')\\
 &  & =\frac{1}{2}\int_{\mathsf{X}^{2}}\left(\left[f(x)-f(x')\right]\left[g(x)-g(x')\right]-\epsilon\left[f(x)+f(x')\right]\left[g(x)+g(x')\right]\right)\eta(dx)\otimes\eta(dx'),\end{eqnarray*}
where the final inequality is due to Fubini's theorem, which is applicable
under the hypotheses of the lemma. The sufficiency part then follows
directly upon setting\begin{eqnarray*}
\epsilon=\frac{\delta}{2+\delta} & \Leftrightarrow & \left(1+\delta\right)=\frac{1+\epsilon}{1-\epsilon}.\end{eqnarray*}
For the necessity part, suppose on the contrary that there exists
$(y,y')\in\mathsf{X}^{2}$ such that \begin{eqnarray*}
\left[f(y)-f(y')\right]\left[g(y)-g(y')\right] & > & \frac{3\delta}{2+\delta}\left[f(y)+f(y')\right]\left[g(y)+g(y')\right],\end{eqnarray*}
 then setting $\eta=\frac{1}{2}\left[\delta_{y}+\delta_{y'}\right]$
and $\epsilon=\frac{\delta}{2+\delta}\geq0$, we obtain

\begin{eqnarray*}
 &  & \int_{\mathsf{X}}\int_{\mathsf{X}}\left(\left[f(x)-f(x')\right]\left[g(x)-g(x')\right]-\epsilon\left[f(x)+f(x')\right]\left[g(x)+g(x')\right]\right)\eta(dx)\eta(dx')\\
 &  & =\int_{\mathsf{X^{2}}}\left(\left[f(x)-f(x')\right]\left[g(x)-g(x')\right]-\epsilon\left[f(x)+f(x')\right]\left[g(x)+g(x')\right]\right)\frac{1}{4}\left[\delta_{y}(dx)\otimes\delta_{y}(dx')+\delta_{y'}(dx)\otimes\delta_{y'}(dx')\right]\\
 &  & +\int_{\mathsf{X^{2}}}\left(\left[f(x)-f(x')\right]\left[g(x)-g(x')\right]-\epsilon\left[f(x)+f(x')\right]\left[g(x)+g(x')\right]\right)\frac{1}{4}\left[\delta_{y}(dx)\otimes\delta_{y'}(dx')+\delta_{y'}(dx)\otimes\delta_{y}(dx')\right]\\
 &  & =-\epsilon\left[f(y)g(y)+f(y')g(y')\right]\\
 &  & +\frac{1}{2}\left(\left[f(y)-f(y')\right]\left[g(y)-g(y')\right]-\epsilon\left[f(y)+f(y')\right]\left[g(y)+g(y')\right]\right)\\
 &  & \geq\frac{1}{2}\left(\left[f(y)-f(y')\right]\left[g(y)-g(y')\right]-3\epsilon\left[f(y)+f(y')\right]\left[g(y)+g(y')\right]\right)>0,\end{eqnarray*}
which completes the proof.
\end{proof}
Note that the sufficient condition is always met, for example, when
$g=\varphi\circ f$ for some positive, strictly decreasing and invertible
function $\varphi$. The factor of $3$ in the necessity part of Lemma
\ref{lem:eta_f_g} arises from considering a suitable $\eta$ with
$N=2$ support points and is illustrative in the case of interest
where each of the measures $\eta_{n,k}^{N}$ is atomic. The following
section further explores the necessity part of Lemma \ref{lem:eta_f_g}
in the case of $\mathsf{X}=\mathbb{R}^{d}$, which is of interest
in the context of \citep{smc:meth:DDJ06}. The purpose of this next
section is to show that in order for (\ref{eq:G_destroys_drift_question})
to hold for any $\eta_{n,k-1}^{N}$, it is necessary that there is
a very specific relationship holds between $G$ and $V$.

\subsection{When $\mathsf{X}=\mathbb{R}^{d}$}

Let $\mathsf{X}=\mathbb{R}^{d}$ and let $\mathcal{B}\mathsf{\left(X\right)}$
be the corresponding Borel $\sigma$-algebra. We denote by $S^{d-1}$
the unit sphere in $\mathbb{R}^{d}$ and for $z$ in the image of
$V$ (resp. $G$), define the contour manifolds

\begin{eqnarray*}
\mathcal{V}_{z} & := & \left\{ x\in\mathbb{R}^{d}:V(x)=z\right\} \quad\text{and}\quad\mathcal{G}_{z}:=\left\{ x\in\mathbb{R}^{d}:G(x)=z\right\} .\end{eqnarray*}
We denote by $B(c,r)$ the closed ball of radius $r$ and centered
at $c$. For $x,x'\in\mathbb{R}^{d}$, let $n(x):=\frac{x}{\left|x\right|}$
and let $\ell_{x'x}$ be the line segment with end points $x$ and
$x'$. Consider the following collection of assumptions.

\begin{condition}\label{hyp:R^d_V_and_G}
\begin{itemize}
\item $G:\mathbb{R}^{d}\rightarrow(0,+\infty)$ and $V:\mathbb{R}^{d}\rightarrow[1,+\infty)$
have continuous first derivatives
\item There exist $r_{0}>0$ and strictly positive and non-decreasing functions
$\phi_{V}:(0,\infty)\rightarrow(0,\infty]$ and $\phi_{G}:(0,\infty)\rightarrow(0,\infty]$
such that whenever $\left|x\right|\geq r_{0}$,\begin{eqnarray}
n(x)\cdot\nabla\log V(x) & \geq & \phi_{V}(\left|x\right|)\quad\text{and}\quad n(x)\cdot\nabla\log G(x)\leq-\phi_{G}(\left|x\right|)\label{eq:V_growth_rate_pre}\end{eqnarray}

\end{itemize}
\end{condition}The conditions of equation (\ref{eq:V_growth_rate_pre})
imply that for all $\delta>0$ and whenever $\left|x\right|\geq r_{0}$,

\begin{eqnarray}
\frac{V\left(x+\delta n(x)\right)}{V(x)} & \geq & \exp\left[\phi_{V}(\left|x\right|)\delta\right],\quad\text{and}\quad\frac{G\left(x+\delta n(x)\right)}{G(x)}\leq\exp\left[-\phi_{G}(\left|x\right|)\delta\right].\label{eq:V_growth_rate_hyp}\end{eqnarray}
The assumptions of (\ref{eq:V_growth_rate_pre}) are largely inspired
by assumptions on probability densities used to verify geometric ergodicity
of certain Metropolis-Hastings kernels \citep{mc:theory:rt96,mc:theory:JH00}
and are realistic in the context of \citep{smc:meth:DDJ06}. It follows
immediately from the arguments of \citet[proof of theorem 2.1]{mc:theory:rt96}
transferred to $G$ and $V$, that for $|x|$ large enough, the contour
manifolds of $G$ and $V$ which contain $x$ are parameterizable
by the unit sphere in the sense that:
\begin{itemize}
\item For each $z>\sup_{\left|y\right|\leq r_{0}}V(y)$ there exists a bijection
between $S^{d-1}$ and $\mathcal{V}_{z}$, and for each $z'<\inf_{\left|y\right|\leq r_{0}}G(y)$
there exists a bijection between $S^{d-1}$ and $\mathcal{G}_{z'}$,
such that, \begin{equation}
\mathcal{V}_{z}=\left\{ w_{z}(\zeta)\zeta\;:\zeta\in S^{d-1}\right\} ,\quad\text{and\quad}\mathcal{G}_{z'}=\left\{ h_{z'}(\zeta)\zeta\;:\zeta\in S^{d-1}\right\} ,\quad,\label{eq:map_S_to_C_V}\end{equation}
where $w_{z}(\cdot)$ and $h_{z}(\cdot)$ are positive and continuous
functions on $S^{d-1}$. Furthermore, $\mathcal{V}_{z}\cap B(0,r_{0})=G_{z'}\cap B(0,r_{0})=\emptyset$.
\end{itemize}
In order to describe the relationship between contour manifolds of
$V$ and $G$ which intersect at some point, we introduce the function
$\psi:\mathbb{R}^{d}\rightarrow[0,\infty]$ defined by\begin{eqnarray}
\psi(x) & := & \sup_{\zeta\in S^{d-1}}\left(\left|h_{G(x)}(\zeta)\right|-\left|w_{V(x)}(\zeta)\right|\right),\label{eq:psi_defn}\end{eqnarray}
which implicitly depends on $G$ and $V$. We have the following proposition. 
\begin{prop}
\label{pro:G_and_V_R^d}Assume $G:\mathsf{X}\rightarrow(0,\infty)$
and $V:\mathsf{X}\rightarrow[1,\infty)$ satisfy (A\ref{hyp:G_bounded-1})
and (A\ref{hyp:R^d_V_and_G}) with\begin{eqnarray}
 &  & \lim_{r\rightarrow\infty}\left(\left[\inf_{s\geq r}\phi_{V}(s)\right]\wedge\left[\inf_{s\geq r}\phi_{G}(s)\right]\right)\inf_{\left|x\right|\geq r}\psi(x)=+\infty.\label{eq:FK_drift_necessary}\end{eqnarray}
 Then for any $0\leq\delta<1$, there exists an atomic $\eta\in\mathcal{P}(\mathsf{X})$
such that \[
\frac{\eta\left(GV\right)}{\eta(G)}>(1+\delta)\eta\left(V\right).\]
\end{prop}
\begin{proof}
In outline, the proof involves showing that if the condition in (\ref{eq:FK_drift_necessary})
holds, then for all $\epsilon\in[0,1$), there exists $(y,y')\in\mathsf{X}^{2}$
such that\[
\left(\frac{G(y)-G(y')}{G(y)+G(y')}\right)\left(\frac{V(y)-V(y')}{V(y)+V(y')}\right)>\epsilon,\]
and then employing Lemma \ref{lem:eta_f_g}. 

Firstly, suppose (\ref{eq:FK_drift_necessary}) fails to hold. Then
fix arbitrarily $\epsilon\in[0,1)$ and let $r\geq r_{0}$ be large
enough that\begin{eqnarray*}
\left(\left[\inf_{s\geq r}\phi_{V}(s)\right]\wedge\left[\inf_{s\geq r}\phi_{G}(s)\right]\right)\inf_{\left|x\right|\geq r}\psi(x) & \geq & 2\log\left(\dfrac{1+\sqrt{\epsilon}}{1-\sqrt{\epsilon}}\right),\end{eqnarray*}
and then let $y$ be any point such that $\left|y\right|>r$, $V(y)>\sup_{\left|u\right|\leq r}V(u)$
and $\psi(y)>0$ (such $r$ and $y$ exist due to the hypotheses of
equation (\ref{eq:V_growth_rate_pre}) and the assumed failure of
(\ref{eq:FK_drift_necessary})). Then recalling the definition of
$\psi$ in (\ref{eq:psi_defn}), and as, for fixed $y$, $h_{G(y)}(\cdot)$
and $w_{V(y)}(\cdot)$ are continuous functions, there exists $\zeta\in S^{d-1}$
such that $\left|h_{G(y)}(\zeta)\right|-\left|w_{V(y)}(\zeta)\right|=\psi(y)>0$.
With a slight abuse, let $x:=h_{G(y)}(\zeta)\zeta$, $x':=w_{V(y)}(\zeta)\zeta$,
and let $y':=\dfrac{1}{2}\left(x+x'\right)$. It follows that the
line segment $\ell_{x'x}$ lies on a ray and $\left|x\right|-\left|y'\right|=\left|y'\right|-\left|x'\right|=\psi(y)/2>0$.
Observe that under the implications of (\ref{eq:V_growth_rate_pre})
stated before the proposition, by construction $V(x)>V(y')>V(x')=V(y)$
and $G(x')>G(y')>G(x)=G(y)$. It must also be the case that $\left|x'\right|>r$,
as otherwise $V(y)=V(x')\leq\sup_{\left|u\right|\leq r}V(u)$, contradicting
the definition of $y$. The situation is illustrated in Figure \ref{fig:contours}.

Due to the hypothesis of equation (\ref{eq:V_growth_rate_hyp}) and
the definition of $r$, we then have

\begin{eqnarray}
\frac{V\left(y\right)}{V(y')}=\frac{V\left(x'\right)}{V(y')} & \leq\exp\left[-\phi_{V}(\left|x'\right|)\dfrac{\psi(y)}{2}\right] & \leq\exp\left[-\frac{1}{2}\left[\inf_{s\geq r}\phi_{V}(s)\inf_{\left|u\right|\geq r}\psi(u)\right]\right]\nonumber \\
 &  & \leq\left(\dfrac{1-\sqrt{\epsilon}}{1+\sqrt{\epsilon}}\right),\label{eq:V_necessary}\end{eqnarray}
and similarly,\begin{eqnarray}
\frac{G\left(y\right)}{G(y')}=\frac{G\left(x\right)}{G(y')} & \leq\exp\left[-\phi_{G}(\left|y'\right|)\dfrac{\psi(y)}{2}\right] & \leq\exp\left[-\frac{1}{2}\left[\inf_{s\geq r}\phi_{G}(s)\inf_{\left|u\right|\geq r}\psi(u)\right]\right]\nonumber \\
 &  & \leq\left(\dfrac{1-\sqrt{\epsilon}}{1+\sqrt{\epsilon}}\right).\label{eq:G_necessary}\end{eqnarray}
It follows from (\ref{eq:V_necessary})-(\ref{eq:G_necessary}) that\begin{eqnarray*}
\frac{1-\dfrac{V(y)}{V(y')}}{1+\dfrac{V(y)}{V(y')}}\geq\sqrt{\epsilon} & \quad\text{and}\quad & \frac{1-\dfrac{G(y)}{G(y')}}{1+\dfrac{G(y)}{G(y')}}\geq\sqrt{\epsilon}.\end{eqnarray*}
The proof is complete upon noting that $\epsilon$ was chosen arbitrarily
in $[0,1)$ and applying the necessity part of Lemma \ref{lem:eta_f_g}.
\end{proof}
\begin{figure}
\begin{minipage}[t]{1\columnwidth}%
\centering\includegraphics[bb=150bp 490bp 470bp 780bp,clip,scale=0.7]{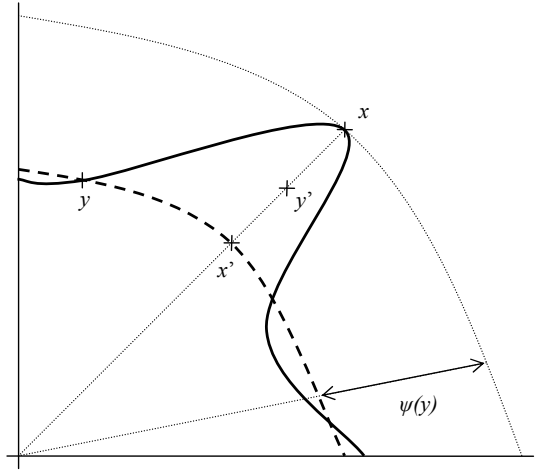}%
\end{minipage}

\caption{Solid line: $\mathcal{G}_{G(y)}$. Dashed line: $\mathcal{V}_{V(y)}$.
Radial distance of $\psi(y)$ outwards from $\mathcal{G}_{G(y)}$
is indicated.\label{fig:contours}}

\end{figure}

We may then formulate the following example.
\begin{cor}
\label{cor:V_and_G_R^d}Let $\mathsf{X}=\mathbb{R}^{2}$. For any
$\epsilon>0$ , let $G$ and $V$ be defined by

\begin{eqnarray*}
V(x)=\exp\left(\left(x_{1}-\epsilon\right)^{2}+x_{2}^{2}\right) & \quad\text{and}\quad & G(x)=\exp\left(-\left[\left(x_{1}+\epsilon\right)^{2}+x_{2}^{2}\right]\right)\end{eqnarray*}
where $x=(x_{1,}x_{2})$. Then for any $0\leq\delta<1$, there exists
an atomic $\eta\in\mathcal{P}(\mathsf{X})$ such that \begin{eqnarray}
\frac{\eta(GV)}{\eta(G)} & > & (1+\delta)\eta(V).\label{eq:cor_V_and_G}\end{eqnarray}
\end{cor}
\begin{proof}
Elementary manipulations show that (A\ref{hyp:R^d_V_and_G}) holds
with $r_{0}=2\epsilon$ and taking $\phi_{G}(\left|x\right|)=\phi_{V}(\left|x\right|)=2(\left|x\right|-\epsilon)$
for $\left|x\right|\geq r_{0}$. For any $r\geq r_{0}$, consider
the contour manifolds $\mathcal{V}_{z}$ and $\mathcal{G}_{z'}$ for
$z=V\left(0,\sqrt{r^{2}-\epsilon^{2}}\right)$ and $z'=G\left(0,\sqrt{r^{2}-\epsilon^{2}}\right)$.
It is straightforward to check that $\psi\left(\left(0,\sqrt{r^{2}-\epsilon^{2}}\right)\right)=2\epsilon$
and the result follows from Proposition \ref{pro:G_and_V_R^d}. 
\end{proof}

This example serves to highlight that a very specific relationship
between the potential functions and the drift function is required
if inequalities of the form (\ref{eq:G_destroys_drift_question})
are to hold for $(1+\delta)\lambda<1$ and for all probability measures.
The application of section \ref{sec:Application} is one situation
where such a relationship holds. We note that it is possible to bound
$\bar{\mathbb{E}}_{\mu}\left[\eta_{n,k}^{N}\left(V\right)\right]$
without confirming (\ref{eq:G_destroys_drift_question}), by appealing
to convexity and the flattening property of the potential functions
combined with the geometric drift; but the application of section
\ref{sec:Application} will not need such an approach and so we do
not report these details here. On the other hand, if inequalities
of the form (\ref{eq:G_destroys_drift_question}) do hold with $(1+\delta)\lambda<1$,
one might then pursue accompanying minorization conditions for the
chain $\left\{ \zeta_{n,k}^{(N)};0\leq k\leq n\right\} $, but it
seems difficult to achieve this in such a way that the minorizing
constants do not degrade as $N$ increases. 

\bibliographystyle{plainnat}
\bibliography{SMC_stability}

\end{document}